\documentclass[aps,prx,reprint,superscriptaddress,amsmath,amssymb,longbibliography]{revtex4-2}
\usepackage[T1]{fontenc}
\usepackage{lmodern}
\usepackage{amsthm}
\usepackage{needspace}
\usepackage{bm}
\usepackage{hyperref}
\hypersetup{
  colorlinks=false,
  pdfborder={0 0 0}
}

\newcommand{\Tr}{\operatorname{Tr}}
\newcommand{\rank}{\operatorname{rank}}
\newcommand{\ran}{\operatorname{ran}}
\newcommand{\supp}{\operatorname{supp}}
\newcommand{\Span}{\operatorname{span}}
\newcommand{\SEP}{\mathrm{SEP}}
\newcommand{\LOCC}{\mathrm{LOCC}}
\newcommand{\LU}{\mathrm{LU}}

\newcommand{\HS}{\mathrm{HS}}
\newcommand{\ket}[1]{|#1\rangle}
\newcommand{\bra}[1]{\langle#1|}
\newcommand{\proj}[1]{|#1\rangle\langle#1|}
\newcommand{\lueq}{\mathrel{\sim_{\LU}}}
\newcommand{\rev}{\mathrel{\longleftrightarrow_{\LOCC}}}

\newtheorem{theorem}{Theorem}
\newtheorem{proposition}{Proposition}

\newtheorem{applemma}{Lemma}[section]

\begin{document}

\title{Complete Classification of Reversible Entanglement Transformations of Two Qutrits}
\author{Wei Song}\email{wsong1@mail.ustc.edu.cn}
\affiliation{School of Physics and Materials Engineering, Hefei Normal University,
Hefei 230601, China}

\author{Xiao-Lan Zong}\email{zxl@hfnu.edu.cn}
\affiliation{School of Physics and Materials Engineering, Hefei Normal University,
Hefei 230601, China}

\author{Ming Yang}
\affiliation{School of Physics, Anhui University, Hefei, 230601, China}
\affiliation{Leibniz International Joint Research Center of Materials Sciences of Anhui Province, Anhui University, Hefei, 230601, China}
\date{\today}

\begin{abstract}
We study single-copy exact deterministic transformations of bipartite quantum states under local operations and classical communication (LOCC), and give a complete classification of the states in two three-level systems that can be transformed into each other. This classification applies to finite protocols and to the countable protocols specified in this paper. If two entangled states can be transformed into each other but do not satisfy local unitary equivalence, then both can be decomposed into an entangled block and a separable part, and the two parts can be distinguished by local measurements. The two entangled blocks have the same probability weight and satisfy local unitary equivalence, while the separable parts may differ. We further prove that every entangled state has a unique least retained block; after retaining its probability weight, the local unitary equivalence class of this block completely determines the reversible equivalence class of the original state. For negative partial transpose states, one only needs to compute the negative spectral subspace of the partial transpose and perform at most three matrix zero tests to determine this retained block.
\end{abstract}

\maketitle

\section{Introduction}

Local operations and classical communication (LOCC) specify the allowed transformations between bipartite quantum states. Two states are said to be mutually interconvertible if each can be transformed into the other exactly and deterministically. For pure states, Nielsen's majorization theorem shows that two states are mutually interconvertible if and only if they are locally unitary (LU) equivalent~\cite{Nielsen}. Therefore, reversible transformations of pure states preserve their Schmidt coefficients.

For mixed states, however, mutually interconvertible states need not be LU equivalent. In some cases, the entangled and separable components can be distinguished locally without disturbing the entangled component. In this case, both parties can retain the entangled part and replace the separable part by local preparation. As long as the target state also has the corresponding flag structure, the original state can be recovered in the same way. Local orthogonality and classical flags provide known examples of such reversible transformations~\cite{LocalOrthogonality,Flags}. However, these examples do not show whether all reversible transformations can be reduced to this mechanism.

The problem of transforming general mixed states is more complex than that for pure states. Gour's results show that for ensemble transformations, and for mixed-state transformations with local dimensions at least four, finitely many conditions are insufficient to give a general characterization~\cite{Gour2005}. Recently, de Vicente studied the relation between mutual interconvertibility and LU equivalence, gave examples of non-unitary reversible transformations, and proved that for states with full entanglement dimensionality, reversible transformations can occur only between LU-equivalent states~\cite{deVicente}. However, in two three-level systems, this result does not yet cover all states with Schmidt number two. Such states can have reduced density matrices of full rank on both sides; hence full local rank alone is not sufficient to draw the above conclusion.

In this paper we prove that in two three-level systems, every pair of reversible states that cannot be related by a whole-state LU transformation has a local flag structure. Specifically, each state can be decomposed into an entangled block and a separable part that can be distinguished by local measurements; the entangled blocks of the two states have the same probability weight and are LU equivalent. The main difficulty of the proof is to show the necessity of this structure. Once this structure is present, bidirectional conversion can be achieved simply by measuring the flag and re-preparing the separable part.

Our proof first selects a best separable approximation compatible with the two conversion channels and proves the rigidity of the entangled residual state obtained after removing the separable part. We then examine the coherence between the product subspace protected by the residual state and its orthogonal complement, thereby determining the part of the original state that must be retained together with the entangled part. The proof requires only that the average output of the protocol equal the target state, not that every measurement branch yield that state.

We prove that every entangled two-qutrit state has a unique least retained block, and that mutual interconvertibility can be determined by comparing these blocks. For negative partial transpose (NPT) states, the least retained block can be determined by spectral analysis and matrix zero tests, after which one only needs to decide whether the two retained blocks are LU equivalent. For general positive partial transpose (PPT) states, the separability testing problem still has to be solved separately. This criterion applies to bidirectional transformations and does not solve the general one-way transformation problem.

\section{Notation and protocol convention}
Throughout this paper, we consider the state space
\(\mathcal H=\mathbb C^3_A\otimes\mathbb C^3_B\).
If there exist completely positive and trace-preserving (CPTP) LOCC channels \(\Lambda\) and \(\Psi\)
such that \(\Lambda(\rho)=\sigma\) and \(\Psi(\sigma)=\rho\),
we write \(\rho\rev\sigma\). The above equalities are all understood in terms of the average output of the protocol.
We allow both finite protocols and countable protocols along the same fixed LOCC protocol tree,
but require that the channels corresponding to their finite prefixes converge in diamond norm~\cite{LOCC}.
The precise definition is given in Appendix~\ref{sec:protocol}.
This convention does not include all channels in the closure of LOCC.
The two parties may use local auxiliary systems, but do not share an entangled catalyst state in advance.
In the following, we denote by \(\Gamma=T_A\) the partial transpose on Alice's subsystem,
and by \(N(\rho)=\Tr(\rho^\Gamma)_-\) the negativity,
where \((\rho^\Gamma)_-\) denotes the negative part of the partial transpose.

\section{Classification of mutually convertible states}

\begin{theorem}\label{thm:classification}
For two-qutrit states \(\rho\) and \(\sigma\),
\(\rho\rev\sigma\) if and only if at least one of the following cases holds.
\begin{enumerate}
\item Both \(\rho\) and \(\sigma\) are separable states.

\item \(\rho\) and \(\sigma\) are locally unitary equivalent, i.e., there exists
\(W=U_A\otimes U_B\) such that \(\sigma=W\rho W^\dagger\).

\item There exist two-dimensional subspaces \(U,U'\subset\mathbb C^3_A\),
unit vectors \(n,n'\) orthogonal to them, respectively,
and a weight \(0<w<1\), such that
\begin{align}
\rho&=w\eta+(1-w)\proj n\otimes\tau_B,\nonumber\\
\sigma&=w\eta'+(1-w)\proj{n'}\otimes\tau'_B.
\label{eq:oneflag}
\end{align}
Here, \(\eta\) and \(\eta'\) are locally unitary equivalent entangled states
supported on \(U\otimes\mathbb C^3\) and
\(U'\otimes\mathbb C^3\), respectively;
\(\tau_B\) and \(\tau'_B\) can be arbitrary quantum states.
The case in which Alice and Bob are exchanged is also included.

\item There exist product subspaces \(L=U\otimes V\) and \(L'=U'\otimes V'\),
where \(U,U',V,V'\) are all two-dimensional local subspaces, such that
\begin{align}
\rho&=w\eta+(1-w)\tau,\nonumber\\
\sigma&=w\eta'+(1-w)\tau',\qquad 0<w<1.
\label{eq:twoflag}
\end{align}
Here, \(\eta\) and \(\eta'\) are locally unitary equivalent entangled states
supported on \(L\) and \(L'\), respectively;
\(\tau\) and \(\tau'\) are separable states
supported on \(L^\perp\) and \((L')^\perp\), respectively.
\end{enumerate}
\end{theorem}

The cases in Theorem~\ref{thm:classification} are not mutually exclusive;
the case $w=1$ has already been included in Case~2.
In Cases~3 and~4, the parties measure the local flags, apply the corresponding
local unitary transformation to the retained block, and re-prepare the target separable part
on the remaining outcomes; the finite protocols are given in Appendix~\ref{sec:shell}.
Exchanging the initial and target states yields the reverse protocol.
Cases~1 and~2 are realized by local preparation and local unitary transformations, respectively.
Therefore, every pair of states in the theorem can be bidirectionally converted via a finite LOCC protocol.

The retained block and the exterior part may each contain internal coherence.
Since the exterior part needs to be re-prepared, its original coherence can be discarded.
The weight $w$ is the probability of obtaining the retained block and is not necessarily equal to the weight of the entangled part in the best separable approximation.
The theorem only states that such a flag protocol exists; it does not restrict which initial measurement must be used by other protocols realizing the same transformation.

Below we briefly outline the proof of necessity.
Detailed arguments are given in the two lemmas in Appendix~\ref{sec:rigiditylemmas} and in Appendices~\ref{sec:protocol}--\ref{sec:shell}.

We still use $\SEP$ to denote the cone of separable positive operators.
Consider the optimization problem in the best separable approximation~\cite{BSA}:
maximize $\Tr S$ subject to $S\in\SEP$ and $0\le S\le\rho$.
This maximum is invariant under reversible transformations.
The round-trip channel $\Psi\Lambda$ maps the compact convex set of optimal solutions into itself,
and hence has a fixed point therein. Choosing this fixed point yields
\begin{equation}
\rho=t\delta+S,\qquad \sigma=t\delta'+S',
\label{eq:compatiblemain}
\end{equation}
where the two given channels realize bidirectional transformations between
$\delta$ and $\delta'$, and between $S$ and $S'$, respectively.
By the optimality of the decomposition, the ranges of both $\delta$ and $\delta'$
contain no nonzero product vector.
Such a subspace is called a completely entangled subspace (CES),
and the corresponding residual state is called a core state.
The above construction does not depend on the uniqueness of the best separable approximation.

We prove that two-qutrit states whose range is a completely entangled subspace,
whenever they are mutually interconvertible, must be locally unitary equivalent.
If the local rank of one of the core states is $(3,3)$,
one can further prove that the two given channels are both
mixtures of product unitary channels on the entire space.
For an NPT core state, the invariance of negativity means that
the positive and negative parts of its partial transpose are mapped respectively
to the corresponding parts of the partial transpose of the target state.
Therefore, the auxiliary state
\begin{equation}
J(\delta)=\frac{(\delta^\Gamma)_-}{N(\delta)}
\label{eq:Jmain}
\end{equation}
can be transformed via the corresponding conjugate protocol.
This protocol is obtained by taking the complex conjugate of Alice's local Kraus operators in the original protocol.
In bidirectional conversion, the entanglement of formation and the negativity
of the core state and of the auxiliary state, respectively, are both invariant.
When the negative spectral subspace is at least two-dimensional,
these equality conditions, together with Lemma~\ref{lem:joint}
and the dimension bound for the kernel of the partial transpose when one local factor is fixed,
jointly rule out local effect operators that are not proportional to the identity operator.
If the negative spectral subspace is one-dimensional and the corresponding vector has Schmidt rank three,
the conclusion follows from the equality condition for the monotonicity of entanglement entropy of pure states.
If the Schmidt rank of this vector is two,
then it protects the product subspace spanned by the Schmidt supports on both sides.
For a CES core state of full local rank,
Lemma~\ref{lem:residual} extends this protection to the entire space.
This argument does not require the partial transpose to be invertible.
Qubit--qutrit core states of rank two and PPT core states of rank four
are handled using matrix pencils and the structure of low-rank PPT states, respectively
~\cite{Pencils,PPT4,LowRank}.

After aligning the two core states by local unitary transformations,
both channels leave the same core state $\delta$ invariant.
A pure core state of Schmidt rank two protects its $2\times2$ Schmidt product subspace:
each Kraus operator, restricted to this subspace,
is a scalar multiple of a product unitary operator, and the sum of the corresponding weights is one.
A mixed $2\times3$ core state protects the $2\times3$ product subspace in which it lies.
If the original state has coherence between the protected subspace and its orthogonal complement,
the Koashi--Imoto coherence relation~\cite{KI}
can enlarge the protected range while keeping the Kraus weights unchanged.
If the enlarged subspace occupies all three local levels on both sides,
the product structure of the Kraus operators forces the two channels to be
mixtures of product unitary channels on the entire space, yielding Case~2.
Otherwise, what is protected is still a proper product subspace.
The weight within this subspace cannot decrease in either the forward or the reverse conversion direction,
so it must be equal in the initial and target states.
Therefore, for the states under consideration, the exterior part does not transfer weight to the retained block,
and the two retained blocks are locally unitary equivalent.

For a $2\times3$ product subspace,
its orthogonal complement has only one dimension on Alice's side,
so the exterior part must be separable.
For a $2\times2$ product subspace,
if the exterior part is still entangled,
one can select from it another CES residual state compatible with the given channel.
Mixing it with the protected pure core state
yields a CES state of full local rank,
and the rigidity of the core state then gives Case~2.
In all remaining cases, the exterior part is separable,
yielding Eqs.~\eqref{eq:oneflag} and~\eqref{eq:twoflag},
in which the flags can be read by both parties on their respective physical subsystems.

For entangled states $\rho$ and $\sigma$,
Cases~2--4 can be equivalently written in the following projector form:
there exist product projectors $R,R'$ and a product unitary operator $W$,
where the local ranks of the two projectors on each side are both two or three, such that
\begin{gather}
[\rho,R]=[\sigma,R']=0,\quad
\rho-R\rho R,\ \sigma-R'\sigma R'\in\SEP,\label{eq:projector1}\\
WRW^\dagger=R',\quad
R'\sigma R'=W(R\rho R)W^\dagger.
\label{eq:projector2}
\end{gather}
Here the projectors select the entire retained block,
and there is no additional requirement that they commute with
the one-dimensional support projector of the chosen pure core state.

\section{The least retained block}
The above projector conditions can be further simplified:
one only needs to compare the retained blocks uniquely determined by the two states.
For an entangled state, a product projector \(R\) whose local ranks on both sides are two or three,
which commutes with the state, and for which the part outside the retained block is separable,
i.e., which satisfies Eq.~\eqref{eq:projector1},
is called an admissible projector of the state.
We prove that, ordered by subspace inclusion,
there exists a unique least element \(R_*(\rho)\) among all admissible projectors. Define
\begin{equation}
C(\rho)=R_*(\rho)\rho R_*(\rho),\qquad
w_*(\rho)=\Tr C(\rho).
\label{eq:C}
\end{equation}

\begin{theorem}\label{thm:criterion}
For entangled two-qutrit states \(\rho\) and \(\sigma\), we have
\begin{equation}
\rho\rev\sigma\quad\Longleftrightarrow\quad
C(\rho)\lueq C(\sigma).
\label{eq:criterion}
\end{equation}
\end{theorem}

The retained block is not normalized, so it simultaneously preserves its probability weight and the information about the state itself.
Therefore, Theorem~\ref{thm:criterion} reduces the determination of reversible transformations
to a comparison between two well-defined retained blocks, without any additional choice of projector.
The proof and the concrete decision method are given in
Appendices~\ref{sec:canonical}--\ref{sec:algorithm}.

For an NPT state, let \(H_-=(\rho^\Gamma)_-\),
and denote by \(P_{\supp X}\) the support projector of a positive operator \(X\). Define
\begin{equation}
P_A=(P_{\supp\Tr_B H_-})^T,\qquad
P_B=P_{\supp\Tr_A H_-} .
\label{eq:negativeplanes}
\end{equation}
Since we use the convention \(\Gamma=T_A\), a transpose is needed on Alice's side.
Any admissible projector must be one of the following projectors:
\begin{equation}
R_0=P_A\otimes P_B,\quad R_A=P_A\otimes I,\quad
R_B=I\otimes P_B,\quad I.
\label{eq:candidates}
\end{equation}
Indeed, the separable exterior part remains positive after partial transpose,
and its support is orthogonal to the transformed retained product subspace.
Therefore, the local subspaces on both sides of the retained block must respectively contain
the subspaces determined by Eq.~\eqref{eq:negativeplanes}.
In a three-dimensional space, a subspace containing a given two-dimensional subspace
can only be that subspace itself or the entire space.

When both \(P_A\) and \(P_B\) have rank two, let
\(X=(P_A\otimes P_B^\perp)\rho(P_A^\perp\otimes P_B)\), and compute
\begin{align}
d_0&=\|R_0\rho(I-R_0)\|_{\HS}^2+\|X\|_{\HS}^2,\nonumber\\
d_A&=\|R_A\rho(I-R_A)\|_{\HS}^2,\nonumber\\
d_B&=\|R_B\rho(I-R_B)\|_{\HS}^2 .
\label{eq:zerotests}
\end{align}
\(d_0=0\) if and only if \(\rho\) has no coherence terms between
the subspace selected by \(R_0\) and its orthogonal complement,
and the exterior part is separable.
In this case, take \(R_*=R_0\).
If \(d_0>0\), then take \(R_*=R_A\) when \(d_A=0\),
and take \(R_*=R_B\) when \(d_B=0\);
if neither \(d_A\) nor \(d_B\) is zero, take \(R_*=I\).
When \(d_0>0\), the two one-sided tests cannot both vanish.
If there are duplicate projectors in Eq.~\eqref{eq:candidates}, only one of them needs to be kept.
After finding the retained blocks of the two states, one still needs to compare whether they are locally unitary equivalent;
in general, having the same spectrum is not sufficient to guarantee this.

For PPT entangled states, the only admissible projector is the identity projector.
If there exists an admissible projector corresponding to a proper product subspace,
then this subspace is at most \(2\times3\).
In this dimension, a PPT state must be separable~\cite{PPTcriterion},
so both the retained block and the exterior part are separable, and the whole state is separable,
contradicting the assumption.
Therefore \(R_*=I\) and \(C(\rho)=\rho\),
and the reversible equivalence class of a PPT entangled state is its entire local unitary orbit.
For a general PPT input state, one still needs to determine whether it is separable;
this problem is not eliminated by the retained-block criterion.

For example, let
\(\ket{\Phi}=(\ket{00}+\ket{11})/\sqrt2\), and take
\begin{align}
\rho&=\tfrac35\proj\Phi+\tfrac25\proj{22},\nonumber\\
\sigma&=\tfrac35\proj\Phi+\tfrac25\proj2\otimes I_3/3.
\label{eq:examplemain}
\end{align}
The least retained blocks of both states are \(C=\tfrac35\proj\Phi\).
Alice tests whether the system is in the third level:
if it is not in the third level, the entangled part is retained, and the probability of this outcome is \(3/5\);
if it is in the third level, then Bob prepares \(I_3/3\) in the forward conversion,
and prepares \(\proj2\) in the reverse conversion.
Both states have local rank three on both sides,
but their global ranks are two and four, respectively, so they are not locally unitary equivalent.
The other examples in Appendix~\ref{sec:examples} further show that,
even if the negative parts of the partial transposes are completely identical,
the retained blocks of two states need not be locally unitary equivalent.

\section{Conclusion}
Within the specified class of LOCC protocols, this paper completely classifies exact deterministic reversible transformations between two-qutrit states.
Apart from global local unitary transformations, reversible transformations are allowed to change only the separable exterior part of the state,
and this part must be distinguishable from the retained entangled block by local measurements.
The probability weight of the entangled block and its local unitary equivalence class remain invariant;
whether there is coherence between it and the orthogonal complement determines how large a block must be retained.

The least retained block uniquely determines the reversible equivalence class to which an entangled state belongs.
For NPT states, determining this block does not require solving the best separable approximation problem;
for general PPT states, separability still has to be decided separately.
When the local dimensions are higher, multiple blocks distinguished by local flags can be entangled simultaneously.
How to completely classify reversible transformations in this case remains an open problem.

\begin{acknowledgments}
This work was supported by the Natural Science Foundation of Anhui Province under Grant No. 2508085ZD001, and Scientific Research Projects of Anhui Provincial Department of Education under Grant Nos. 2025AHGXZK20126 and 2025AHGXZK50054.
\end{acknowledgments}

\appendix
\setcounter{secnumdepth}{2}
\numberwithin{equation}{section}
\renewcommand{\theequation}{\Alph{section}\arabic{equation}}
\numberwithin{proposition}{section}
\renewcommand{\theHsection}{appendix.\Alph{section}}
\renewcommand{\theHsubsection}{appendix.\Alph{section}.\arabic{subsection}}
\renewcommand{\theHequation}{appendix.\Alph{section}.\arabic{equation}}

\section{Two lemmas on core-state rigidity}\label{sec:rigiditylemmas}

\subsection{A joint-kernel condition}
The following lemma is used to prove the rigidity of core states,
in order to rule out local effect operators that are not proportional to the identity operator.
This conclusion does not require the partial transpose to be invertible.
Let \(P_\pm(H)\) denote the projectors onto the positive and negative spectral subspaces of \(H\).

\begin{applemma}\label{lem:joint}
Let a positive operator \(\eta\ge0\) have range a completely entangled subspace
and local rank three on Alice's side.
Let a Hermitian operator \(D\) have exactly one negative eigenvalue, one zero eigenvalue,
and one positive eigenvalue. Then the following two conditions cannot hold simultaneously:
\begin{align}
P_{\supp\eta}(D\otimes I)P_{\supp\eta}&=0,\nonumber\\
P_+(\eta^\Gamma)(D^T\otimes I)P_-(\eta^\Gamma)&=0.
\label{eq:jointconditions}
\end{align}
\end{applemma}

\begin{proof}
Let \(H=\eta^\Gamma\) and \(X=D^T\otimes I\).
By the first condition, \(\eta(D\otimes I)\eta=0\).
In a basis that diagonalizes \(D\), write the operator in block form according to Alice's basis;
expanding yields
\begin{equation}
0=\Tr[\eta(D\otimes I)\eta(D\otimes I)]=\Tr(HXHX).
\label{eq:traceidentity}
\end{equation}
Equation~\eqref{eq:traceidentity} is the trace identity obtained from the block expansion.
In the eigenbasis of \(H\), the right-hand side can be written as
\(\sum_{ij}\lambda_i\lambda_j|X_{ij}|^2\).
By the second condition in Eq.~\eqref{eq:jointconditions},
all terms with \(\lambda_i\) and \(\lambda_j\) of opposite signs vanish.
The remaining terms are nonnegative, and the total sum is zero;
hence \(X_{ij}=0\) whenever \(\lambda_i\lambda_j\ne0\).
Therefore \(HXH=0\).
This derivation also applies to the case where \(H\) has a zero eigenvalue.

Write \(D=\operatorname{diag}(-a,0,b)\), where \(a,b>0\),
and take the decomposition \(\eta=RR^\dagger\), where \(R\) is a full-column-rank matrix of rank \(r\):
\begin{equation}
R=\begin{pmatrix}R_0\\R_1\\R_2\end{pmatrix},\qquad
aR_0^\dagger R_0=bR_2^\dagger R_2.
\label{eq:Rfactor}
\end{equation}
If there exists a nonzero vector \(x\) satisfying \(R_0x=0\),
then the above equation also gives \(R_2x=0\).
In this case \(Rx=\ket1\otimes R_1x\) is a nonzero product vector in the range of \(\eta\),
contradicting the CES assumption.
Hence \(R_0\) is injective, \(r\le3\),
and there exists a unitary operator \(U\) on Bob's side such that
\(R_2=\sqrt{a/b}\,UR_0\).

Let \(V=\ran R_0\) and \(G=R_0R_0^\dagger\),
and write \(R_1=TR_0\), where \(T:V\to\mathbb C^3\).
From the individual blocks of \(HXH=0\) we obtain
\begin{equation}
[U,R_0R_i^\dagger R_jR_0^\dagger]=0
\qquad(0\le i,j\le2).
\label{eq:commutators}
\end{equation}
Taking \((i,j)=(0,0)\) gives \([U,G^2]=0\),
and hence \([U,G]=0\), so \(V\) is invariant under \(U\).
Next let \(T_0=P_VT\) and \(T_1=P_{V^\perp}T\).
Since \(G\) is invertible on \(V\),
the equations corresponding to \((i,j)=(0,1)\) and \((1,1)\) yield
\begin{equation}
[U|_V,T_0]=[U|_V,T_1^\dagger T_1]=0.
\label{eq:Tcommutators}
\end{equation}

If \(U|_V\) is proportional to the identity operator, then \(R_2\propto R_0\),
contradicting local rank three on Alice's side.
Therefore \(r=1\) is also impossible.

If \(r=3\), then \(V=\mathbb C^3\), and hence \(T_1=0\).
Since \(T\) commutes with \(U\),
one can choose an eigenvector of \(T\) within some eigenspace of \(U\).
It is also an eigenvector of \(U\).
Take its preimage \(x\) under \(R_0\);
then \(R_0x\), \(R_1x\), and \(R_2x\) are collinear with one another,
so \(Rx\) is a nonzero product vector.

If \(r=2\), then the two eigenvalues of \(U|_V\) are distinct;
otherwise it would be proportional to the identity operator.
Equation~\eqref{eq:Tcommutators} shows that
\(T_0\) and \(T_1^\dagger T_1\) are both diagonal in the eigenbasis of \(U|_V\).
Since \(\rank T_1\le1\), at least one basis vector lies in \(\ker T_1\).
This vector is also an eigenvector of \(T_0\),
and hence an eigenvector of \(T=T_0+T_1\).
Taking its preimage under \(R_0\) again yields a nonzero product vector.

Both cases above contradict the assumption that the range of \(\eta\) is a CES.
Therefore the two conditions cannot hold simultaneously.
\end{proof}

\subsection{One negative vector of Schmidt rank two}
The following lemma deals with the cases that the preceding local effect operator argument could not rule out.
The proof uses the results on protected subspaces from Appendix~\ref{sec:protected}.

\begin{applemma}\label{lem:residual}
Let the range of \(\delta\) be a completely entangled subspace, with local rank \((3,3)\), and suppose that
\begin{equation}
\delta^\Gamma=H_+-\nu\proj v,\qquad
\nu>0,\quad \operatorname{SR}(v)=2 .
\label{eq:singlenegative}
\end{equation}
Then any separable CPTP channel that leaves \(\delta\) invariant
is a mixture of product unitary channels on the entire space.
Moreover, in any product Kraus representation of this channel,
every nonzero operator is proportional to some product unitary operator.
\end{applemma}

\begin{proof}
Let \(\Omega\) be a separable CPTP channel that leaves \(\delta\) invariant.
By the equality condition for the invariance of negativity,
\(\widetilde\Omega=\Gamma\Omega\Gamma\) leaves \(\proj v\) invariant.
Taking the complex conjugate of the Schmidt subspace on Alice's side
yields the protected product subspace \(L=U\otimes V\) of the original channel \(\Omega\).
We further extend the protection to
\(E=\Span\{\delta^kL:k\ge0\}\).

If the local rank of \(E\) is \((3,3)\), the conclusion follows from global rigidity.
If the local rank is \((2,3)\), the protection extends to
\(T=U\otimes\mathbb C^3\).
In this case, if \([\delta,P_T]=0\), then, since Alice's side has full local rank,
\(\delta\) must have a nonzero block on \(U^\perp\otimes\mathbb C^3\).
All vectors in this subspace are product vectors, contradicting the CES assumption.
Therefore, the protection must extend beyond \(T\),
which again yields global rigidity.
Exchanging the two parties handles the case of local rank \((3,2)\).

It remains only to rule out \(E=L\).
In this case, \(\delta=\delta_L\oplus\delta_Q\),
where \(Q=L^\perp\) and the range of \(\delta_Q\ne0\) is still a CES.
Choosing a suitable local basis,
one can decompose \(Q\) into two wing subspaces \(U\otimes m\), \(n\otimes V\)
and a corner subspace \(\mathbb C(nm)\).
Take a full-column-rank factorization
\begin{equation}
\delta_Q=RR^\dagger,\qquad
R=\begin{pmatrix}A\\B\\c\end{pmatrix}.
\label{eq:wingfactor}
\end{equation}
If \(A\) or \(B\) has a nonzero kernel vector,
then under the action of \(R\) this vector becomes a nonzero product vector,
contradicting that \(\ran R\) is a CES.
Therefore, both \(A\) and \(B\) are injective,
and \(1\le r=\rank\delta_Q\le2\).

The negative spectral subspace of \(H=\delta^\Gamma\) is spanned by the vector \(v\) in \(L\),
so the compression of \(H\) onto \(Q\) is positive.
Ordering the basis as corner, first wing, second wing, we have
\begin{equation}
H_Q=
\begin{pmatrix}
\|c\|^2&f_A^\dagger&f_B^\dagger\\
f_A&\bar A A^T&0\\
f_B&0&BB^\dagger
\end{pmatrix}\ge0,
\label{eq:wingPT}
\end{equation}
where \(f_A=\bar A c^T\) and \(f_B=B\bar c^{\,T}\).
By the injectivity of \(A\) and \(B\),
\begin{equation}
f_A^\dagger(\bar A A^T)^+f_A
=f_B^\dagger(BB^\dagger)^+f_B=\|c\|^2 .
\label{eq:Schur}
\end{equation}
Therefore, the Schur complement of Eq.~\eqref{eq:wingPT} gives
\(-\|c\|^2\ge0\), and hence \(c=0\).

Since \(nm\perp v\) and \(H\) is positive on \(v^\perp\),
\(\bra{nm}H\ket{nm}=0\) implies \(H\ket{nm}=0\).
In the above local basis, for \(i,j=0,1\) we have
\(\langle ij|H|22\rangle=(BA^\dagger)_{ji}\),
and hence \(BA^\dagger=0\).
Also, since \(A\) has full column rank, \(A^\dagger\) is surjective,
so \(B=0\), contradicting the injectivity of \(B\).
Therefore, \(E=L\) is impossible.
\end{proof}

\section{Protocols and equality conditions}\label{sec:protocol}

Appendices~\ref{sec:protocol}--\ref{sec:shell} give
the proof of Theorem~\ref{thm:classification}:
first derive the equality conditions for reversible protocols, then prove the rigidity of CES core states,
and finally decompose the original state into blocks that can be distinguished by local measurements.
Appendices~\ref{sec:canonical}--\ref{sec:examples} prove
Theorem~\ref{thm:criterion} and explain how to apply this criterion.
The proof of core-state rigidity also uses Lemma~\ref{lem:joint} and Lemma~\ref{lem:residual}
in Appendix~\ref{sec:rigiditylemmas}.

All references are listed in the common bibliography at the end of the paper.
The local rank of a subspace means the rank of the operator obtained by
taking the partial trace of its orthogonal projector over the other party.
Convex mixtures of product unitary channels are called
mixed local unitary (mixed-LU) channels.
The separable positive operators mentioned in the text need not be normalized.

\subsection{Protocol convention and finite-prefix equalities}

The countable protocols considered in this paper proceed step by step along the same fixed LOCC protocol tree,
and each local operation may have finitely or countably many outcomes.
The average channels \(\Lambda_n\) corresponding to its finite prefixes
must converge in diamond norm to some CPTP channel \(\Lambda\):
\begin{equation}
\|\Lambda_n-\Lambda\|_\diamond\longrightarrow0.
\label{eq:channelconvergence}
\end{equation}
For a truncated protocol, we keep the outputs of branches that have terminated and discard the classical records of the other branches. We assume fixed finite-dimensional input and output spaces.
The finite prefixes are taken from a single LOCC protocol. We do not include limits of arbitrary sequences of finite-round LOCC protocols~\cite{LOCC}. Convergence of the output for a fixed input state does not imply the channel convergence in Eq.~\eqref{eq:channelconvergence}.

Local operations can be refined so that each outcome corresponds to a single
Kraus operator. The additional classical records can then be discarded.
Even if local auxiliary systems are used, each Kraus operator maps the original
three-dimensional space into a subspace of dimension at most three.
Using local isometries, we can apply the following argument
on the support of the state.
For countable protocols, the fixed finite-dimensional spaces used by the finite-prefix channels
and the coordinates of the various branches must all be explicitly specified in the protocol.
One cannot adjust the readout separately for different inputs or different truncation positions
and then infer channel convergence from this.

Let \(M\) be the entanglement of formation \(E_F\) or the negativity \(N\).
In the finite-dimensional case, both are continuous convex functions
and satisfy strong monotonicity under LOCC~\cite{Vidal,Negativity}.
If \(\rho\) and \(\sigma\) are mutually interconvertible, then \(M(\rho)=M(\sigma)\).
Fix the depth \(m\) of the protocol tree,
and denote the branch ensemble at this level by
\(\{p_\alpha,\rho_\alpha\}_{\alpha\in\mathcal T_m}\).
For any \(n\ge m\), we have
\begin{equation}
M(\rho)\ge\sum_{\alpha\in\mathcal T_m}p_\alpha M(\rho_\alpha)
\ge M(\Lambda_n\rho)\longrightarrow M(\sigma).
\label{eq:prefixequality}
\end{equation}
Since \(M(\rho)=M(\sigma)\), we let \(n\to\infty\) at fixed \(m\) and obtain
\[
\sum_{\alpha\in\mathcal T_m}p_\alpha M(\rho_\alpha)=M(\rho).
\]
Therefore, at every node reached with positive probability,
the local operations cause no average loss of \(M\).
The second inequality need not be an equality at finite \(n\),
since discarding the classical records may decrease \(M\).
This conclusion does not guarantee that any individual conditional state can be recovered separately,
nor does it require every terminal branch to output the same mixed target state \(\sigma\).

\subsection{Jordan parts and the auxiliary state}

The invariance of negativity also yields a pair of mutually convertible auxiliary states.
Writing the product Kraus operators of the original channel as \(A_i\otimes B_i\),
we have that \(\widetilde\Lambda=\Gamma\Lambda\Gamma\) is a CPTP channel
whose Kraus operators are \(\bar A_i\otimes B_i\).
Taking the complex conjugate of Alice's operators at each node of the protocol
yields a protocol that still has the original structure.
Since this transformation is a linear map on a finite-dimensional channel space,
it also preserves the convergence of channel sequences.

Decompose \(H=\rho^\Gamma\) as \(H=H_+-H_-\).
Since \(\widetilde\Lambda\) is trace-preserving,
the positive operators \(\widetilde\Lambda(H_+)\) and \(\widetilde\Lambda(H_-)\)
have the same traces as \(H_+\) and \(H_-\), respectively.
The invariance of negativity is equivalent to equality in the triangle inequality
for the trace norm of the difference of these two positive operators;
and equality holds if and only if their supports are mutually orthogonal.
Therefore,
\begin{equation}
\widetilde\Lambda(H_\pm)
=\bigl[(\Lambda\rho)^\Gamma\bigr]_\pm .
\label{eq:Jordantransport}
\end{equation}
If the local operations cause no average loss of negativity,
the above argument can be applied separately to each measurement outcome~\cite{Guo}.
For NPT initial and target states,
this shows that the auxiliary state defined by \(J(\rho)=H_-/N(\rho)\)
can also be bidirectionally converted via the corresponding conjugate protocol.
Therefore, Eq.~\eqref{eq:prefixequality} also applies to
the entanglement of formation \(E_F\) and the negativity \(N\) of the auxiliary state.
Its branch probabilities are
\begin{equation}
q_\alpha=p_\alpha\frac{N(\rho_\alpha)}{N(\rho)}.
\label{eq:auxprobability}
\end{equation}
It can be seen that the same branch may occur with different probabilities
for the original state and for the auxiliary state.
Before proving rigidity, one cannot directly set \(q_\alpha=p_\alpha\).

Moreover, the negative spectral subspace \(S_-\) contains no nonzero product vector,
and is therefore a completely entangled subspace.
This is because any product vector satisfies
\(\bra{a,b}\rho^\Gamma\ket{a,b}
=\bra{\bar a,b}\rho\ket{\bar a,b}\ge0\),
while any nonzero vector in \(S_-\) has a negative expectation value with respect to \(\rho^\Gamma\).

\subsection{The local contact equation and a kernel bridge}

Suppose the range of \(\rho\) is a completely entangled subspace,
and its local rank on Alice's side is three.
Consider a node in the protocol that preserves \(E_F\) on average,
and denote by \(F=M^\dagger M\) the effect operator corresponding to one of Alice's measurement outcomes.
Choose an optimal pure-state decomposition of the input state.
Since the average value of the entanglement of formation is invariant,
both the convex-roof inequality and the monotonicity inequality for the entanglement entropy of pure states must be equalities.
Therefore, every pure-state component with nonzero weight also makes the corresponding entropy inequality an equality.
Using the strict concavity of entropy on Bob's side, one obtains~\cite{Guo}
\begin{equation}
P_{U_j}FP_{U_j}=q_jP_{U_j},
\label{eq:contact}
\end{equation}
where \(U_j\) is the Schmidt support on Alice's side of the \(j\)-th pure-state component,
and \(q_j\) is the probability that this component yields this measurement outcome.
In the Schmidt basis, let \(\lambda_j\) be the positive diagonal matrix formed by the squares of the nonzero Schmidt coefficients;
then the invariance of the conditional reduced state can be written as
\(\sqrt{\lambda_j}(P_{U_j}FP_{U_j})^T\sqrt{\lambda_j}
=q_j\lambda_j\).
Since \(\lambda_j\) is invertible on its support,
canceling \(\sqrt{\lambda_j}\) on both sides yields Eq.~\eqref{eq:contact}.
If \(q_j=0\), the same conclusion follows directly from positivity.

If some pure-state component has Schmidt rank three, then \(F=qI\) follows immediately.
Otherwise, since the range of \(\rho\) contains no product vector,
all pure-state components have Schmidt rank two.
In three-dimensional space, these two-dimensional Schmidt subspaces pairwise have nonzero intersection,
so all the \(q_j\) must be equal; denote this common value by \(q\).
Suppose \(D=F-qI\ne0\).
By the eigenvalue interlacing theorem, and excluding the positive semidefinite and negative semidefinite cases,
we find that \(D\) has exactly one negative eigenvalue, one zero eigenvalue, and one positive eigenvalue.
Here, if \(D\) were positive semidefinite or negative semidefinite,
Eq.~\eqref{eq:contact} would require all \(U_j\) to be contained in \(\ker D\),
contradicting the full local rank on Alice's side.

Moreover, since \(DU_j\subset U_j^\perp\),
and \(U_j^\perp\) is only one-dimensional,
each pure-state component, after the action of \(D\otimes I\),
becomes either a product vector or the zero vector. Therefore,
\begin{equation}
(D\otimes I)\rho(D\otimes I)\in\SEP.
\label{eq:Dsep}
\end{equation}
Let \(H=\rho^\Gamma\) and \(X=D^T\otimes I\).
Taking the partial transpose of the above equation gives \(XHX\ge0\).
The equality condition for negativity under local operations further requires that
the supports of the positive and negative parts of \(H\) after the conjugate operation are mutually orthogonal, i.e.,
\(P_+(F^T\otimes I)P_-=0\),
and hence \(P_+XP_-=0\).
Therefore, for any \(z\in S_-\),
\(Xz\) has no component in the positive spectral subspace.
Combined with \(\bra{Xz}H\ket{Xz}\ge0\),
it also has no component in the negative spectral subspace,
and can only lie in the kernel of \(H\).
Thus,
\begin{equation}
XS_-\subset\ker H,\qquad P_-XP_-=0.
\label{eq:kernelbridge}
\end{equation}
That is, if the local effect operator is not proportional to the identity operator,
then the \(X\) constructed from it must map the negative spectral subspace \(S_-\)
into the kernel of the partial transpose.
For operations on Bob's side, the same conclusion is obtained by taking \(X=I\otimes D\).
These conditions depend only on the equality conditions on the given input state;
there is no need to assume that the operation also preserves the corresponding entanglement measure for other states.

\subsection{Mutual mixed LU reduces to LU}

Let \(\sigma=\sum_i p_iW_i\rho W_i^\dagger\),
and suppose that the reverse conversion is also implemented by a mixed local unitary channel.
By the convexity of the squared Hilbert--Schmidt norm,
the forward conversion gives \(\|\sigma\|_{\HS}^2\le\|\rho\|_{\HS}^2\),
while the reverse conversion gives the opposite inequality; hence the two are equal. Therefore,
\begin{equation}
\sum_i p_i\|W_i\rho W_i^\dagger-\sigma\|_{\HS}^2
=\|\rho\|_{\HS}^2-\|\sigma\|_{\HS}^2=0.
\label{eq:variance}
\end{equation}
Since every term on the left-hand side is nonnegative, for every \(p_i>0\) we have
\(W_i\rho W_i^\dagger=\sigma\),
i.e., \(\rho\lueq\sigma\);
see Lemma~3 of Ref.~\cite{deVicente}.
The same argument also applies to positive operators with equal traces,
and hence can also be used for the comparison of unnormalized retained blocks below.
It should be emphasized that both the forward and reverse conversion channels must be mixed local unitary channels.

\section{Protected subspaces and product Kraus operators}
\label{sec:protected}

\subsection{A pure fixed point protects its Schmidt rectangle}

If a separable channel leaves a pure state of Schmidt rank two invariant,
the trace-preserving condition restricts the action of the Kraus operators
on the product subspace of the Schmidt supports on both sides.
In the following, this subspace is called the Schmidt product subspace
(or Schmidt rectangle).

\begin{applemma}\label{lem:pureprotection}
Let a separable CPTP channel \(\Omega\) leave a pure state \(\proj\psi\) invariant,
where \(\psi\) has Schmidt rank two,
and the Schmidt subspaces on the two sides are \(U\) and \(V\), respectively.
Then, for any finite product Kraus representation of this channel,
the operators on \(L=U\otimes V\) all satisfy
\begin{equation}
\begin{gathered}
K_i|_L=\sqrt{p_i}\,W_i,\\
W_i\text{ is a product unitary operator on }L,\\
\sum_i p_i=1.
\end{gathered}
\label{eq:protected}
\end{equation}
When \(p_i=0\), the action of \(K_i\) on \(L\) is zero.
Moreover, \(\Omega(P_L)=P_L\).
\end{applemma}

\begin{proof}
Since the output is the pure state \(\proj\psi\),
every Kraus operator, when acting on \(\psi\), must yield a vector proportional to \(\psi\).
Hence \(K_i\psi=c_i\psi\).
Let \(p_i=|c_i|^2\); by the trace-preserving property, \(\sum_i p_i=1\).

For terms with \(p_i>0\),
the local factors of \(K_i=A_i\otimes B_i\) map \(U\) and \(V\) back into themselves, respectively.
Writing \(K_i\psi=c_i\psi\) as an equality of coefficient matrices in the Schmidt basis
and then taking determinants gives
\(|\det(A_i|_U)\det(B_i|_V)|=p_i\).
Applying the arithmetic--geometric mean inequality to the squares of the four singular values of \(K_i|_L\),
and using the trace-preserving condition, we obtain
\begin{equation}
4=\sum_i\|K_iP_L\|_{\HS}^2
\ge4\sum_{p_i>0}p_i=4.
\label{eq:tracefour}
\end{equation}
Both sides are equal, so every positive-weight term must make the above inequality an equality,
and all zero-weight terms must have zero norm.
Therefore, the positive-weight terms are all scalar multiples of product unitary operators on \(L\),
while the zero-weight terms vanish identically on all of \(L\).
\end{proof}

If the pure state under consideration is the normalized negative part of the partial transpose,
then, when transferring the conclusion from the conjugate channel back to the original channel,
one needs to take the complex conjugate of the Schmidt subspace on Alice's side.

\subsection{Coherence extends protection with the same weights}

The following lemma extends the protection from a given subspace to
other subspaces coupled to the channel-invariant state.
The extension does not change the Kraus representation or the weights of the individual terms.
This will later be used to derive properties of the two conversion channels
from the properties of the round-trip channel.

\begin{applemma}\label{lem:extension}
Let \(\Omega(\omega)=\omega\ge0\),
and suppose that its Kraus representation satisfies Eq.~\eqref{eq:protected} on \(L\),
where the unitary operators need not be of product form.
Then the protection can be extended with the original weights to
\begin{equation}
E=\Span\{\omega^kL:k\ge0\},\qquad
K_i|_E=\sqrt{p_i}\,\widehat W_i ,
\label{eq:extension}
\end{equation}
where \(\widehat W_i\) are unitary operators on \(E\).
When \(p_i=0\), \(K_i\) vanishes identically on \(E\).
Here \(L\subset\supp\omega\) is not required.
\end{applemma}

\begin{proof}
\(\Theta=\omega+P_L\) is invariant under the channel
and is strictly positive on its support \(S=\supp\Theta\).
Therefore, each \(K_i\) leaves \(S\) invariant.
Let \(E_0=L\).
Suppose the protection has been extended to \(E_k\), and that
\(\Omega(P_{E_k})=P_{E_k}\).
With respect to the decomposition \(S=E_k\oplus Q_k\), write
\(K_i|_S=\left(\begin{smallmatrix}C_i&B_i\\0&D_i\end{smallmatrix}\right)\).
Examining the block of \(\Omega(\Theta)=\Theta\) on \(Q_k\)
and using the trace-preserving condition, we obtain
\begin{equation}
\sum_i\Tr\bigl(B_i\Theta_{Q_k}B_i^\dagger\bigr)=0.
\label{eq:faithfulnoflow}
\end{equation}
Since \(\Theta_{Q_k}\) is strictly positive on \(Q_k\),
all \(B_i\) must vanish.
This is precisely the step in Lemma~2 of Koashi--Imoto that rules out inflow from outside~\cite{KI}.
Therefore, each \(K_i|_S\) is block diagonal with respect to \(E_k\oplus Q_k\).

Take the polar decomposition
\begin{equation}
\begin{gathered}
X_k=P_{Q_k}\omega P_{E_k}=V_kN_k,\\
F_k=\supp N_k,\quad R_k=\ran X_k.
\end{gathered}
\label{eq:polar}
\end{equation}
Since \(L\subset E_k\), the corresponding off-diagonal block of \(\Theta\) is also equal to \(X_k\).
The operator \(X_k+X_k^\dagger\) is invariant under the channel,
and the positive and negative parts in its Jordan decomposition are also invariant,
so their supports are invariant subspaces of each \(K_i\).
The support of the positive part is the graph subspace
\(\{f+V_kf:f\in F_k\}\).
Since \(K_i\) is block diagonal with respect to \(E_k\oplus Q_k\),
the invariance of this graph subspace implies that \(K_iF_k\subset F_k\) and
\begin{equation}
K_i|_{R_k}=V_k(K_i|_{F_k})V_k^\dagger .
\label{eq:matching}
\end{equation}
Alternatively, one may apply Theorem~2 of Koashi--Imoto
(Eq.~(22) in the preprint) to the Stinespring isometry
\(\mathcal Vx=\sum_i K_ix\otimes\ket i\),
and then compare the coefficients of the environment basis vectors \(\ket i\) one by one to obtain the same conclusion.
The entire derivation uses the same Kraus representation.

When \(p_i>0\), the finite-dimensional unitary operator
\(K_i|_{E_k}/\sqrt{p_i}\) leaves \(F_k\) invariant,
and its restriction to \(F_k\) is still a unitary operator.
By Eq.~\eqref{eq:matching},
the restriction of \(K_i\) to \(R_k\) is also \(\sqrt{p_i}\) times a unitary operator.
When \(p_i=0\), this equation gives \(K_i|_{R_k}=0\).
Therefore, the protection is extended with the original weights to
\(E_{k+1}=E_k\oplus R_k=E_k+\omega E_k\).
Since the space is finite-dimensional, the iteration stops growing after finitely many steps,
and one finally obtains the subspace \(E\) in Eq.~\eqref{eq:extension}.
Moreover, \(E\) is an invariant subspace of the Hermitian operator \(\omega\),
and hence \([P_E,\omega]=0\).
\end{proof}

This lemma gives a conclusion for a fixed Kraus representation.
To obtain from it a local flag that both parties can read,
one still needs to use the product form of the Kraus operators.
The tensor factors in the abstract Koashi--Imoto decomposition
do not necessarily correspond to the actual Alice--Bob bipartition.

\subsection{Qutrit product rigidity and one-sided completion}

For product Kraus operators, the local dimensions determine how the protection can be extended.
Each three-level system has only one dimension left outside a two-dimensional Schmidt subspace,
so the extension of protection either covers a product subspace that is the full three-dimensional space on one side,
or makes every Kraus operator proportional to a product unitary operator on the entire space.

\begin{applemma}\label{lem:productrigidity}
Let \(L=U\otimes V\) be a \(2\times2\) product subspace,
and suppose that product Kraus operators satisfy Eq.~\eqref{eq:extension} on \(E\supset L\).
If the local rank of \(E\) is \((3,3)\),
then every nonzero Kraus operator is proportional to a product unitary operator on the entire space.
If the local rank is \((2,3)\),
then the protection extends with the same weights to \(T=U\otimes\mathbb C^3\);
the \((3,2)\) case is handled by exchanging the two parties.
If coherence makes the protection extend further beyond \(T\), then global rigidity follows.
\end{applemma}

\begin{proof}
For terms with \(p_i>0\), by suitably adjusting the scalar factor between the two local factors,
we can write
\begin{equation}
K_i=\sqrt{p_i}
\begin{pmatrix}U_i&x_i\\0&\alpha_i\end{pmatrix}
\otimes
\begin{pmatrix}V_i&y_i\\0&\beta_i\end{pmatrix},
\label{eq:triangular}
\end{equation}
where \(U_i,V_i\) are two-dimensional unitary operators.
Let \(n,m\) span the orthogonal complements of \(U,V\), respectively.
In the quotient space modulo \(L\),
ordered as \(Um,nV,nm\), and after removing the factor \(\sqrt{p_i}\), we obtain
\begin{equation}
\begin{pmatrix}
\beta_iU_i&0&\beta_i x_i\\
0&\alpha_iV_i&\alpha_i y_i\\
0&0&\alpha_i\beta_i
\end{pmatrix}.
\label{eq:quotient}
\end{equation}
The restriction of this operator to \(E\ominus L\) is a unitary operator.

First suppose that the local rank of \(E\) is \((3,3)\).
If the projection of \(E\ominus L\) onto the corner \(nm\) is nonzero,
then the linear functional extracting the corner component is a left eigenfunctional,
with corresponding eigenvalue \(\alpha_i\beta_i\).
By unitarity, \(|\alpha_i\beta_i|=1\).
If the corner projection is zero, then full local rank requires that
the projections of \(E\ominus L\) onto the two wings are both nonzero.
Taking the Hilbert--Schmidt norm of the unitary intertwining relations satisfied by these two projections gives
\(|\alpha_i|=|\beta_i|=1\).
Therefore, in either case, \(|\det K_i|^{2/9}=p_i\), and hence
\begin{equation}
9=\sum_i\Tr K_i^\dagger K_i
\ge9\sum_{p_i>0}|\det K_i|^{2/9}=9.
\label{eq:tracenine}
\end{equation}
Both sides are equal, so all singular values of every positive-weight term must be equal,
and every zero-weight term must vanish on the entire space.
Combined with the product structure, both local factors are also proportional to unitary operators.
Here we have used the determinant--trace inequality from Ref.~\cite{GheorghiuGour}
and proved that equality holds for the present Kraus representation.

Now suppose that the local rank of \(E\) is \((2,3)\).
In this case, \(E\ominus L\) is a nonzero subspace of \(U\otimes m\).
From the invariance of this subspace and the vanishing of inflow into \(L\), we obtain \(y_i=0\);
from the preservation of the norm, we obtain \(|\beta_i|=1\).
Therefore, all positive-weight terms are proportional to product unitary operators on \(T\).
These terms already contribute the entire trace \(6\sum_i p_i=6\),
so the trace-preserving condition requires the zero-weight terms to vanish identically on \(T\).

If the protection extends further beyond \(T\),
the newly added part is a nonzero subspace of \(n\otimes\mathbb C^3\).
In this case, the positive-weight terms can be written as
\(\sqrt{p_i}\left(\begin{smallmatrix}U_i&x_i\\0&\alpha_i\end{smallmatrix}\right)
\otimes V_i\), where \(V_i\) is a three-dimensional unitary operator.
Since the newly added part cannot flow into \(T\), we must have \(x_i=0\);
and from the preservation of the norm on the newly added subspace, we obtain \(|\alpha_i|=1\).
Finally, the trace-preserving condition again requires the zero-weight terms to vanish identically on the entire space.
\end{proof}

\subsection{Passing rigidity back to the given channels}

Let \(\{L_a\}\) and \(\{M_b\}\) be finite product Kraus representations of the given channels \(\Lambda\) and \(\Psi\), respectively.
Apply the rigidity lemma to the Kraus representation \(\{M_bL_a\}_{b,a}\) of the composite channel.
If global rigidity holds, then \(K_{ba}=M_bL_a\) satisfies
\(K_{ba}^\dagger K_{ba}=K_{ba}K_{ba}^\dagger=t_{ba}I\).
Summing over \(b\) and using the trace-preserving property of \(\Psi\), we obtain
\begin{equation}
L_a^\dagger L_a=\Bigl(\sum_b t_{ba}\Bigr)I.
\label{eq:factorback}
\end{equation}
Therefore, \(\Lambda\) is a mixed local unitary channel, and \(\Lambda(I)=I\).
Using this property and summing over \(a\), we obtain
\(M_bM_b^\dagger=(\sum_a t_{ba})I\).
Thus, both given channels are mixed local unitary channels.
If the coefficient in Eq.~\eqref{eq:factorback} is zero,
then \(L_a\) vanishes identically on the entire space.
For \(M_b\), the same conclusion follows using \(\Lambda(I)=I\).
Hence we obtain not only the local unitary equivalence of the initial and target states,
but also the structure of the two channels themselves.

Next consider the corresponding statement on subspaces.
Suppose that both channels leave the same pure state of Schmidt rank two invariant,
and that the composite channel protects the subspace \(T=U\otimes\mathbb C^3\),
where \(T\) contains the Schmidt product subspace of this pure state. In this case,
\begin{equation}
P_TL_a^\dagger L_aP_T=q_aP_T.
\label{eq:subfactorback}
\end{equation}
If \(q_a>0\), then \(L_a\psi\ne0\).
Since \(\proj\psi\) is a fixed point of the channel,
the Alice factor of \(L_a\) must map \(U\) back into itself.
Combining this with the tensor product structure in Eq.~\eqref{eq:subfactorback},
we find that the restriction of the Alice factor to \(U\) and the action of the Bob factor on \(\mathbb C^3\)
are both proportional to unitary operators.
Zero-weight terms vanish on \(T\), so \(\Lambda(P_T)=P_T\).
This gives
\(M_bP_TM_b^\dagger=\sum_aK_{ba}P_TK_{ba}^\dagger=r_bP_T\),
so \(\Psi\) also protects \(T\).

For countable protocols, the closedness of the finite-dimensional separable cone guarantees that
the Choi operator of the limit channel remains separable.
By Carath\'eodory's theorem, this channel can then be written in a finite product Kraus representation.
This representation is used only for algebraic derivations concerning separable channels,
and does not imply that the channel can be implemented by a finite-round LOCC protocol.

\section{Rigidity of completely entangled cores}\label{sec:core}

We now prove the rigidity of CES states.
When the local rank is full, one can further determine that
the two given conversion channels are both mixed local unitary channels.
Appendix~\ref{sec:shell} will use this conclusion
to extend the analysis of core states to original states containing a separable part.

\begin{proposition}[Rigidity of core states and structure of the conversion channels]
\label{prop:core}
For CES states \(\delta\) and \(\delta'\),
\(\delta\rev\delta'\) if and only if \(\delta\lueq\delta'\).
If one of the states has local rank \((3,3)\),
then the given forward and reverse conversion channels
are both mixed local unitary channels acting on the entire space.
\end{proposition}

\subsection{Possible core ranks}

The separable weight of a state in the best separable approximation is zero
if and only if its range contains no nonzero product vector, i.e., its range is a CES.
If a product vector exists in the range,
one can subtract a sufficiently small positive multiple of the projector onto this vector from the state
so that the remainder is still a positive operator.
Conversely, if there exists a nonzero separable operator \(S\) satisfying \(0\le S\le\delta\),
then the range of \(\delta\) necessarily contains a product vector.
By the dimension bound for CESs~\cite{CES}, we have
\begin{equation}
\rank\delta\le(r_A-1)(r_B-1).
\label{eq:CESbound}
\end{equation}
Combined with the relevant results on low-rank PPT states~\cite{PPTcriterion,LowRank},
this shows that the core state has only four possibilities:
a pure state of Schmidt rank two in a \(2\times2\) space;
a rank-two NPT state in a \(2\times3\) or \(3\times2\) space;
a PPT state with local rank three on both sides and global rank four;
and an NPT state with local rank three on both sides and global rank at most four.

\subsection{Rank-two qubit--qutrit cores}

We first prove that for a rank-two CES state in a \(2\times3\) space,
its partial transpose is invertible on the corresponding six-dimensional space.
Take the two vectors spanning its support;
their coefficient matrices form a matrix pencil of constant rank two.
In the Kronecker decomposition~\cite{Pencils}, only one \(L_2\) block can occur:
a regular block would drop rank at some projective parameter;
shorter singular blocks would either produce a rank-one matrix
or make the dimensions of the local supports insufficient.
Therefore, after an invertible local transformation,
a positive operator supported on this subspace can be written as
\begin{equation}
\begin{gathered}
\delta_0=RGR^\dagger,\\
R=(\ket{00}+\ket{11},\,\ket{01}+\ket{12}),\\
G=\begin{pmatrix}a&c\\d&b\end{pmatrix}>0,\quad d=\bar c.
\end{gathered}
\label{eq:pencilnormal}
\end{equation}
A direct computation gives
\begin{equation}
\delta_0^\Gamma=
\begin{pmatrix}
a&c&0&0&0&0\\
d&b&0&a&c&0\\
0&0&0&d&b&0\\
0&a&c&0&0&0\\
0&d&b&0&a&c\\
0&0&0&0&d&b
\end{pmatrix},\qquad
\det\delta_0^\Gamma=(ab-cd)^3>0 .
\label{eq:pencildeterminant}
\end{equation}
When \(G=I\), \(\delta_0^\Gamma\) has four positive eigenvalues and two negative eigenvalues.
Since the cone of positive definite matrices is connected
and the above determinant is always nonzero, the numbers of positive and negative eigenvalues cannot change.
Invertible local transformations correspond to invertible congruence transformations of the partial transpose,
so the partial transpose of the original core state is also invertible on the corresponding six-dimensional space.
Here invertible local transformations are used only to determine this spectral property;
local unitary equivalence is deduced from the equality conditions below.

By Eq.~\eqref{eq:contact},
the effect operator on the qubit side must be proportional to the identity operator on its two-dimensional support.
If the effect operator on the three-level side is not proportional to the identity operator,
then, by the counterpart of Eq.~\eqref{eq:kernelbridge} on Bob's side,
the operator \(I\otimes D\) must map the negative spectral subspace into the kernel of the partial transpose.
However, the partial transpose is invertible on this six-dimensional space,
so \(I\otimes D\) must map the entire negative spectral subspace to zero.
Since \(\ker D\) is only one-dimensional,
every nonzero vector in \(\mathbb C^2\otimes\ker D\) is a product vector,
and this space cannot contain any nonzero vector from the negative spectral CES subspace.
This rules out the above possibility.
Therefore, in every finite-prefix branch with positive probability,
the core state undergoes only local isometries,
and these isometries can be extended to local unitary transformations on the entire three-level space.

Next we determine the action of the given channel on the entire product subspace.
Using the local unitary equivalence proved below, we align the core states
so that the channel leaves the core state on \(T=U\otimes\mathbb C^3\) invariant.
In any finite prefix, the product Kraus operators of positive-probability branches
have the following properties: the Bob factor is proportional to a unitary operator,
and the Alice factor is proportional to an isometry on \(U\).
For Kraus terms that are themselves nonzero but have zero probability on the core state,
we must have \(AP_U=0\).
Indeed, before the first zero-probability outcome occurs,
the core state undergoes only local isometries;
and since the local support on Bob's side has full rank,
a nonzero Bob effect operator cannot yield a zero-probability outcome.

Normalize the Hilbert--Schmidt norms of the nonzero local factors to one,
and define the compact set \(\mathcal K=\mathcal K_+\cup\mathcal K_0\):
\begin{equation}
\begin{aligned}
\mathcal K_+&=\{A\otimes B:\|A\|_{\HS}=\|B\|_{\HS}=1,\\
&\qquad B^\dagger B=I_3/3,\\
&\qquad P_UA^\dagger AP_U=aP_U,\ a\ge0\},\\
\mathcal K_0&=\{A\otimes B:\|A\|_{\HS}=\|B\|_{\HS}=1,\\
&\qquad AP_U=0\}.
\end{aligned}
\label{eq:compactK}
\end{equation}
Let \(\mathcal C(\Lambda_n)\) denote the unnormalized Choi operator;
then \(\mathcal C(\Lambda_n)/9\) belongs to
the compact convex hull of
\(\{\ket{K}\!\rangle\langle\!\bra{K}:K\in\mathcal K\}\).
By channel convergence and Carath\'eodory's theorem,
the limit channel also has a finite product Kraus representation satisfying the same conditions.

The limit channel leaves the core state invariant,
and each Kraus output term is a positive operator,
so the support of every nonzero output term must lie in the range of the core state.
Its image space on Alice's side can therefore only be \(U\).
Hence, the positive-weight terms are all scalar multiples of product unitary operators on \(T\),
and the sum of their weights is one.
Zero-weight terms vanish identically on \(T\),
but may still have nonzero action on the exterior space,
and hence cannot be directly deleted from the channel.

\subsection{Full-local-rank PPT cores}

For rank-four PPT entangled two-qutrit states,
we use the following two structural properties:
the state generates an extreme ray of the PPT cone
[Ref.~\cite{LowRank}, Theorem~3(iv) and Definition~9],
and its projective local stabilizer group is finite
[Ref.~\cite{PPT4}, Proposition~27].
Let a separable channel leave this state invariant.
Since the operator obtained after the action of each product Kraus operator still satisfies the PPT condition,
the extreme-ray property implies that every output term must be proportional to the original state, i.e.,
\begin{equation}
K_i\delta K_i^\dagger=p_i\delta,\qquad
p_i\ge0,\qquad\sum_i p_i=1.
\label{eq:PPTextreme}
\end{equation}
When \(p_i>0\), since both local ranks of \(\delta\) are full,
\(A_i\) and \(B_i\) must be invertible.
The finiteness of the projective local stabilizer group further implies that
some positive integer power of each local factor is proportional to the identity operator.
Applying Eq.~\eqref{eq:PPTextreme} repeatedly
and then taking determinants of the local factors, we obtain
\begin{equation}
p_i=|\det A_i\det B_i|^{2/3}=|\det K_i|^{2/9}.
\label{eq:PPTdet}
\end{equation}
Thus, Eq.~\eqref{eq:tracenine} gives global rigidity,
and at the same time guarantees that the zero-weight terms vanish identically on the entire space.
Using Eq.~\eqref{eq:factorback},
this conclusion can then be transferred separately to the two given conversion channels.

\subsection{A kernel bound with one local factor fixed}

Even if an NPT core state has full local rank, its partial transpose may still be singular.
Below we prove that, after fixing the local vector on one side,
the dimension of the kernel of the partial transpose in the corresponding subspace is at most one.
When the dimension of the negative spectral subspace is at least two,
combined with Eq.~\eqref{eq:kernelbridge},
this conclusion rules out local effect operators that are not proportional to the identity operator.

\begin{applemma}\label{lem:fiber}
Let a positive operator \(\omega\) have range a CES and full local rank.
Then for any nonzero vector \(a\),
\begin{equation}
\dim\bigl(\ker\omega^\Gamma\cap(a\otimes\mathbb C^3)\bigr)\le1 .
\label{eq:fiber}
\end{equation}
The same dimension upper bound holds when the local vector on Bob's side is fixed.
\end{applemma}

\begin{proof}
Choose a suitable local basis and handle the conjugation according to the partial transpose convention
so that \(a=e_2\).
Suppose there exists a two-dimensional subspace \(W\) such that
\(e_2\otimes W\subset\ker\omega^\Gamma\).
From the expectation-value relation of the partial transpose on product vectors
and the positivity of \(\omega\),
we obtain \(\omega(e_2\otimes W)=0\).
Writing \(\omega\) as a block matrix according to Alice's basis,
the above two kernel conditions give
\(\omega_{i2}W=\omega_{2i}W=0\).
Since \(\omega_{2i}=\omega_{i2}^\dagger\),
both the left and right supports of \(\omega_{i2}\) lie in the one-dimensional space
\(W^\perp=\mathbb Cz\),
so it can be written as \(\omega_{i2}=c_i\proj z\).
Full local rank on Alice's side guarantees \(c_2>0\).
Therefore,
\(\omega\ket{2,z}=(\sum_i c_i\ket i)\otimes z\)
is a nonzero product vector in the range of \(\omega\),
contradicting the assumption that the range is a CES.
\end{proof}

\subsection{Full-local-rank NPT cores}

Let \(\eta=J(\delta)\) and \(k=\dim S_-(\delta^\Gamma)\).
Suppose that in Alice's operation there appears the first effect operator that is not proportional to the identity operator.
From the contact equation, one can construct an operator \(D\) with inertia \((-,0,+)\),
and it satisfies Eq.~\eqref{eq:kernelbridge}.
Below we discuss separately according to the local rank of the auxiliary state on Alice's side.

If \(r_A(\eta)=3\),
then Eq.~\eqref{eq:kernelbridge} shows that
\(\eta\) and \(D^T\) satisfy the first condition of Lemma~\ref{lem:joint}.
The invariance of the negativity of the auxiliary state under the conjugate protocol gives the second condition of that lemma.
However, these two conditions cannot hold simultaneously, yielding a contradiction.

If \(r_A(\eta)=2\),
then by the equality condition for the entanglement of formation of the auxiliary state,
on its two-dimensional local support \(U\) we have \(P_UF^TP_U=q'P_U\).
Since the range of \(\eta\) contains no product vector,
the Schmidt support on Alice's side of every pure-state component in its optimal decomposition must be \(U\).
The eigenvalue interlacing theorem further requires that \(q'\) equal the middle eigenvalue of \(F^T\),
namely \(q\), and hence
\begin{equation}
(D^T\otimes I)S_-\subset U^\perp\otimes\mathbb C^3 .
\label{eq:fibermap}
\end{equation}
\(\ker D^T\) is only one-dimensional,
so every nonzero vector in \(\ker D^T\otimes\mathbb C^3\) is a product vector.
Since \(S_-\) is a CES, the above map must be injective on \(S_-\).
When \(k\ge2\), the dimension of its image is also at least two.
Combining Eqs.~\eqref{eq:kernelbridge} and~\eqref{eq:fibermap},
this contradicts the dimension upper bound given by Lemma~\ref{lem:fiber}.
The local rank of the auxiliary state also cannot be one, since otherwise its range would contain a product vector.

For Bob's side, it suffices to adjust the treatment of the transpose accordingly, and the argument is completely the same.
Therefore, when \(k\ge2\), the local effect operators on both sides must be proportional to the identity operator.
When \(k=1\), if the negative spectral subspace is spanned by a vector of Schmidt rank three,
then the auxiliary state is pure, and the same conclusion follows from the equality condition for its entanglement of formation.
If the Schmidt rank of this vector is two,
then apply Lemma~\ref{lem:residual} to the separable channel \(\Psi\Lambda\) that leaves \(\delta\) invariant,
and then use Eq.~\eqref{eq:factorback} to transfer the conclusion to the two conversion channels.
Here we directly use the fact that the round-trip channel leaves \(\delta\) invariant,
without assuming that some conditional branch can be reversed separately.

At this point, all cases have been covered.
When \(k\ge2\), the local rank of the auxiliary state on both sides can only be two or three;
when \(k=1\), the Schmidt rank of the negative eigenvector can only be two or three.
What is discussed here is the Schmidt rank of the negative eigenvector,
which should not be identified with the Schmidt number of \(\delta\).
The above argument always allows \(\delta^\Gamma\) to be a singular operator.

\subsection{Endpoints and actual infinite protocols}

In the cases where the above local rigidity holds,
every finite-prefix branch state with positive probability is locally unitary equivalent to the initial state.
Since the convex hull of an LU orbit in a finite-dimensional space is compact,
the limiting target state also belongs to this convex hull,
and can therefore be written as a finite mixture of states LU equivalent to the initial state.
Denote the channel realizing this mixture by \(\Phi\).

If the initial state has full local rank,
then every nonzero local effect operator on the entire three-level space
is proportional to the identity operator.
At the same time, from \(F\ge0\), \(\delta_A>0\), and \(\Tr(F\delta_A)=0\),
we obtain \(F=0\); the same conclusion holds on Bob's side.
Therefore, the local operators corresponding to zero-probability outcomes must be zero,
and leave no unconstrained action outside the range of the core state.
By induction step by step, we obtain that
every finite-prefix channel is a mixed local unitary channel on the entire input space.
The convex hull of the Choi operators of product unitary channels is compact,
so the limiting channel still lies in it.
This proves that the given forward channel is a mixed local unitary channel.

To handle the given reverse channel,
first realize the finite mixture \(\Phi\) by shared randomness,
and then execute the original reverse protocol, with the subsequent operations not using this random label.
The resulting protocol proceeds along the same fixed protocol tree,
and its average final output is exactly the initial state.
By Eq.~\eqref{eq:Jordantransport},
the corresponding conjugate protocol also leaves \(J(\delta)\) invariant.
Therefore, both the original state and the auxiliary state satisfy the finite-prefix equality conditions.
After applying \(\Phi\), every conditional state is locally unitary equivalent to \(\delta\).
By induction step by step along the reverse protocol tree,
we find that every nonzero local effect operator must be proportional to the identity operator.
Since the branch states always have full local rank,
any subsequent nonzero local effect operator corresponds to a positive-probability outcome.
Therefore, this argument determines the structure of the given reverse channel on the entire input space.
By the compactness of the set of Choi operators of mixed local unitary channels,
the limiting reverse channel is still a mixed local unitary channel.
The above construction consists of a finite preparation followed by a countable protocol tree,
which conforms to the protocol convention of this paper,
and does not involve concatenating two already completed infinite protocols.

For the cases handled through the fixed-point property of separable channels,
it suffices to apply Eq.~\eqref{eq:factorback} directly.
Then Eq.~\eqref{eq:variance} yields the LU equivalence of the two endpoint states.
As long as one of the states has full local rank, the argument can start from that state,
without the need to determine the local rank of the other state in advance.

If neither of the two mixed core states has full local rank,
then, by the argument for the qubit--qutrit case,
they respectively belong to the convex hulls of each other's LU orbits,
and Eq.~\eqref{eq:variance} can still be used to obtain LU equivalence.
However, the convex hull of the LU orbit of a rank-two state contains no pure state,
so a mixed core state cannot be mutually convertible with a pure core state.

Finally, if both core states are pure states of Schmidt rank two,
then their binary Schmidt entropies are equal.
Within the interval determined after ordering the Schmidt coefficients,
the binary entropy is strictly monotonic, so the two sets of Schmidt coefficients are identical,
and the two states are locally unitary equivalent.
For finite protocols, this is also a direct consequence of Nielsen's theorem
~\cite{Nielsen}.
Thus, Proposition~\ref{prop:core} is proved.

\section{From core rigidity to local flags}
\label{sec:shell}

\subsection{A compatible optimal decomposition}

To apply the rigidity of core states to the original states,
one needs to choose a best separable approximation compatible with the given conversion channels,
so that the two residual states can also be mutually converted by these two channels.
An arbitrary choice of a pair of optimal decompositions does not guarantee this.
Define the maximum separable weight and the corresponding set of optimal solutions as
\begin{align}
b(\rho)&=\max\{\Tr S:0\le S\le\rho,\ S\in\SEP\},\nonumber\\
\mathcal F_\rho&=\{S:0\le S\le\rho,\ S\in\SEP,\ \Tr S=b(\rho)\}.
\label{eq:BSAset}
\end{align}
\(\mathcal F_\rho\) is a nonempty compact convex set.
A separable channel maps any feasible operator \(S\) of the initial state
to a feasible operator of the target state without changing its trace.
Therefore, if \(\rho\) and \(\sigma\) are mutually interconvertible,
then \(b(\rho)=b(\sigma)=b\)~\cite{BSA}.

The round-trip channel \(\Omega=\Psi\Lambda\) maps \(\mathcal F_\rho\) into itself.
Choose any \(S_0\in\mathcal F_\rho\) and construct the averaged sequence
\begin{equation}
\bar S_n=\frac1n\sum_{k=0}^{n-1}\Omega^k(S_0),\qquad
\|\Omega(\bar S_n)-\bar S_n\|_1\le\frac{2b}{n}.
\label{eq:Cesaro}
\end{equation}
By compactness, one can extract a convergent subsequence.
The above inequality also guarantees that its limit \(S_*\) satisfies \(\Omega(S_*)=S_*\).
If \(\rho\) is entangled, then \(t=1-b>0\). Let
\begin{equation}
S'=\Lambda(S_*),\qquad
\delta=\frac{\rho-S_*}{t},\qquad
\delta'=\frac{\sigma-S'}{t}.
\label{eq:compatible}
\end{equation}
The ranges of both residual states are CESs.
Otherwise, one could subtract a nonzero separable positive operator from the residual part,
thereby increasing the separable weight, contradicting optimality.
By the above construction, the given channels satisfy
\begin{equation}
\Lambda(\delta)=\delta',\quad\Psi(\delta')=\delta,\qquad
\Lambda(S_*)=S',\quad\Psi(S')=S_*.
\label{eq:fourtransport}
\end{equation}
Therefore, there exists at least one pair of optimal decompositions compatible with the given channels.
Other optimal decompositions need not be compatible, and uniqueness of the decomposition need not be assumed.
Here only the composition of the two channels is used,
and there is no need to actually concatenate two already completed infinite protocols.

By Proposition~\ref{prop:core}, we have \(\delta'=W\delta W^\dagger\),
where \(W\) is a product unitary operator.
Rewriting the target state as \(W^\dagger\sigma W\),
and simultaneously changing the forward and reverse channels to
\(W^\dagger\Lambda(\,\cdot\,)W\) and
\(\Psi(W\,\cdot\,W^\dagger)\), respectively,
makes both channels leave the same core state \(\delta\) invariant.

\subsection{Weight conservation and no inflow}

Suppose both channels protect the subspace corresponding to the projector \(P\).
That is, each Kraus operator is a scalar multiple of a unitary operator on this subspace,
the sum of the corresponding weights is one, and zero-weight terms vanish identically on this subspace.
Let \(Q=I-P\). Then the Kraus operators of \(\Lambda\) can be written as
\begin{equation}
K_i=\begin{pmatrix}C_i&R_i\\0&D_i\end{pmatrix},
\qquad
\Lambda^\dagger(P)-P=
\begin{pmatrix}0&0\\0&\sum_i R_i^\dagger R_i\end{pmatrix}\ge0 .
\label{eq:inflow}
\end{equation}
This shows that the weight in the protected subspace cannot decrease in either conversion direction.
Since the conversion is reversible, this weight can only remain unchanged,
i.e., \(\Tr P\rho=\Tr P\sigma=w\).
Therefore, every nonnegative contribution from external inflow must vanish:
\(R_iQ\rho Q R_i^\dagger=0\).
Moreover, by the positivity of \(\rho\),
\(\ran(Q\rho P)\subset\ran(Q\rho Q)\),
and hence \(R_iQ\rho P=0\).

Thus, the retained block, the exterior part, and the coherence block between them
are mapped respectively to the corresponding parts in the target state.
If one state has no coherence terms between the protected subspace and its orthogonal complement,
neither does the other.
The action of both channels on the retained block is a mixed local unitary transformation,
so by Eq.~\eqref{eq:variance},
these two positive operators with the same weight must be locally unitary equivalent.

\subsection{Necessity of the four cases}

Separable states can only be mutually interconvertible with separable states,
corresponding to Case~1 of Theorem~\ref{thm:classification}.
For entangled states, we first choose core states compatible with the given channels,
align them by local unitary transformations,
and then discuss separately according to the local supports of the core states.

If the core state has local rank \((3,3)\),
then Proposition~\ref{prop:core} guarantees that the two original channels
are both mixed local unitary channels on the entire space.
Applying Eq.~\eqref{eq:variance} to \(\rho,\sigma\) then yields Case~2.

If the core state is a mixed state of local rank \((2,3)\),
then, by the channel structure results of Appendix~\ref{sec:core},
both channels protect \(T=U\otimes\mathbb C^3\),
and their composite channel leaves \(\rho\) invariant.
When \([\rho,P_T]\ne0\),
Lemma~\ref{lem:extension} extends the protection beyond \(T\),
and Lemma~\ref{lem:productrigidity} then gives global rigidity;
combining this with Eq.~\eqref{eq:factorback} yields Case~2.
When \([\rho,P_T]=0\),
Eq.~\eqref{eq:inflow} shows that the retained blocks of the two states
have the same weight and are locally unitary equivalent.
The exterior part lies in \(U^\perp\otimes\mathbb C^3\);
since Alice's side is only one-dimensional, it must be separable.
This yields Case~3; when the retained weight is one, it belongs to Case~2.
Exchanging the two parties handles the case of local rank \((3,2)\).

If the core state is a pure state \(\proj\psi\) of Schmidt rank two,
then Lemma~\ref{lem:pureprotection} guarantees that
both channels protect its Schmidt product subspace \(L\).
Applying Lemma~\ref{lem:extension} to the composite channel that leaves \(\rho\) invariant,
the protection can be extended to \(E=\Span\{\rho^kL:k\ge0\}\).
If the local rank of \(E\) is \((3,3)\), we obtain Case~2.
If the local rank is \((2,3)\) or \((3,2)\),
then the protection further covers a product subspace that is the full three-dimensional space on one side.
By Eq.~\eqref{eq:subfactorback},
the two given channels also protect this subspace respectively.
Then, following the argument in the previous paragraph,
we obtain Case~2 or Case~3.

The remaining possibility is \(E=L\).
In this case, \(\rho\) has no coherence terms between \(L\) and its orthogonal complement,
i.e., \([\rho,P_L]=0\).
Equation~\eqref{eq:inflow} guarantees that the two retained blocks have the same weight
and are locally unitary equivalent; the two exterior parts are also mutually converted by the original channels.
When the exterior part is separable, we obtain Case~4.
If the exterior part is still entangled, normalize it,
and apply Eqs.~\eqref{eq:BSAset}--\eqref{eq:fourtransport} to this reversible pair of states.
In this way, one can construct compatible CES states \(\epsilon,\epsilon'\) in \(L^\perp\),
which are also mutually converted by the original two channels.
Therefore,
\begin{equation}
\zeta=\tfrac12\proj\psi+\tfrac12\epsilon,\qquad
\zeta'=\tfrac12\proj\psi+\tfrac12\epsilon'
\label{eq:combinedcore}
\end{equation}
are also mutually converted under these two channels.

The ranges of \(\zeta\) and \(\zeta'\) are still CESs.
Taking \(\zeta\) as an example, suppose that
there exists a nonzero product vector in
\(\mathbb C\psi\oplus\ran\epsilon\).
Since \(L\) is a product subspace,
the projection of this vector onto \(L\) is still a product vector or the zero vector.
But this projection must also be proportional to the entangled vector \(\psi\),
and hence can only be zero.
Therefore, the original vector belongs to \(\ran\epsilon\),
contradicting the assumption that the range of \(\epsilon\) is a CES.
The case of \(\zeta'\) is analogous.

Both states also have full local rank.
If \(\epsilon\) does not occupy Alice's third level,
then, since \(\ran\epsilon\subset L^\perp\),
we must have \(\ran\epsilon\subset U\otimes m\).
Every nonzero vector in this subspace is a product vector, so it cannot accommodate a CES.
Therefore, \(\epsilon\) must occupy Alice's third level;
the same holds on Bob's side.
Combining this with the Schmidt supports of \(\psi\) on both sides,
we find that the local rank of \(\zeta\) is \((3,3)\),
and the same conclusion holds for \(\zeta'\).
Thus, applying Proposition~\ref{prop:core} to \(\zeta,\zeta'\),
we conclude that the two original channels are both mixed local unitary channels on the entire space,
which again falls into Case~2.

So far, necessity has been proved.
The separability of the exterior part follows from the above channel argument.
One cannot conclude that an operator remains separable after projection by \(I-P_L\)
merely because the operator itself is separable,
since this projection does not necessarily preserve separability.

\subsection{Sufficiency and the projector formulation}

Sufficiency is given by the following finite protocols.
In Case~1, the input state is discarded,
and the target state is prepared according to its finite product-state decomposition.
In Case~2, the given local unitary transformation is applied directly.
In Case~3, Alice performs the measurement \(\{P_U,\proj n\}\).
If the first outcome is obtained, \(\eta\) is retained with probability \(w\),
and then transformed into \(\eta'\) by a local unitary transformation;
if the other outcome is obtained, \(\proj{n'}\otimes\tau'_B\) is prepared.

In Case~4, Alice first tests whether the system lies in \(U\);
if so, Bob then tests whether his system lies in \(V\).
The branch in which both parties obtain a positive outcome corresponds exactly to the retained block,
and its probability is \(w\).
On this branch, the corresponding local unitary transformation is applied;
on the remaining branches, the two parties use classical communication
to re-prepare \(\tau'\) according to its finite product-state decomposition.
Since the exterior part is re-prepared,
destroying its internal coherence during the measurement does not affect the obtainment of the target state.
Exchanging the initial and target states yields the corresponding reverse protocol.
Moreover, LOCC cannot transform a separable state into an entangled state.
This completes the proof of Theorem~\ref{thm:classification}.

We now supplement the equivalence between the projector conditions and the above classification.
If both states are entangled,
then in Cases~2--4, taking the projectors corresponding to the retained subspaces
that are the entire space, restricted on one side, or restricted on both sides, respectively,
yields Eqs.~\eqref{eq:projector1} and~\eqref{eq:projector2}.
Conversely, if these two projector conditions hold,
then the retained block must be nonzero and entangled.
Let \(w=\Tr R\rho R\).
When \(w=1\), this corresponds to Case~2;
when \(0<w<1\),
a retained projector of local rank \((2,3)\) or \((3,2)\) corresponds to Case~3,
and a retained projector of local rank \((2,2)\) corresponds to Case~4.
The local unitary transformation relating the two entangled retained states
also maps their corresponding two-dimensional local supports onto each other.
Therefore, the retained projector refers to the entire retained state,
and does not exclude the possibility that this state still has coherence
with respect to the support of some chosen pure core state.

\section{The least retained block}\label{sec:canonical}

\subsection{Admissible projections are determined by the negative part}

We now prove that there exists a unique least element among the admissible retained projectors.
For an entangled state \(\rho\), a product projector \(R=P\otimes Q\)
whose ranks on both sides are two or three and which satisfies
\begin{equation}
[\rho,R]=0,\qquad T_R=(I-R)\rho(I-R)\in\SEP,
\label{eq:admissible}
\end{equation}
is called an admissible projector.
In this case, \(\rho\) decomposes into the retained block \(C_R=R\rho R\)
and the separable exterior part \(T_R\), i.e., \(\rho=C_R+T_R\).
Since \(\rho\) is entangled, \(C_R\) must be nonzero and entangled.
The identity projector is clearly an admissible projector.

Let \(\rho\) be an NPT state.
For any admissible projector \(R\),
write the exterior part \(T_R\) as a positive linear combination of product projectors.
Every product vector \(a\otimes b\) therein is orthogonal to the support of \(R\),
so necessarily \(Pa=0\) or \(Qb=0\).
After taking the partial transpose, the corresponding projectors remain positive,
and their supports are orthogonal to the support of \(P^T\otimes Q\).
Therefore,
\begin{equation}
\rho^\Gamma=C_R^\Gamma\oplus T_R^\Gamma,\qquad
(\rho^\Gamma)_-=(C_R^\Gamma)_- .
\label{eq:negativeblock}
\end{equation}
That is, the negative part of \(\rho^\Gamma\) comes entirely from the retained block.
Consequently, the local support projectors in Eq.~\eqref{eq:negativeplanes} satisfy
\(P_A\le P\) and \(P_B\le Q\).
The negative spectral subspace contains no nonzero product vector,
so the ranks of both \(P_A\) and \(P_B\) are at least two.
In a three-dimensional space, however, a subspace containing a given two-dimensional subspace
can only be that subspace itself or the entire space.
Hence, Eq.~\eqref{eq:candidates} covers all admissible projectors.
When computing the local supports, it suffices to know the negative spectral subspace;
its projector can be used in place of the negative part weighted by eigenvalues.

\subsection{Separability of the exterior}

If the retained subspace is restricted on only one side,
then its orthogonal complement has only one dimension on that side, and the exterior part is therefore necessarily separable.
For a \(2\times2\) retained subspace,
the separability of the exterior part is determined by the following lemma:
there must be no coherence terms between the two wings,
but coherence between the wings and the corner may be retained.

\begin{applemma}\label{lem:exterior}
Let \(P_A,P_B\) be rank-two projectors, and let
\(T\ge0\) be supported on \((P_A\otimes P_B)^\perp\). Then
\begin{equation}
T\in\SEP\quad\Longleftrightarrow\quad
(P_A\otimes P_B^\perp)T(P_A^\perp\otimes P_B)=0.
\label{eq:exteriorzero}
\end{equation}
\end{applemma}

\begin{proof}
Ordering the blocks as the two wings and the corner, \(Um,nV,nm\), write \(T\) as
\begin{equation}
T=\begin{pmatrix}A&X&a\\X^\dagger&B&b\\a^\dagger&b^\dagger&c\end{pmatrix}.
\label{eq:exteriormatrix}
\end{equation}
If \(T\) is separable, then every product vector \(x\otimes y\) in its separable decomposition
satisfies \(\|P_Ax\|^2\|P_By\|^2=0\).
That is, either \(x\parallel n\) or \(y\parallel m\).
Therefore, the block between the two wings of every product projector is zero,
and after summing we still have \(X=0\).

Conversely, if \(X=0\), then from \(T\ge0\) we obtain
\begin{equation}
\begin{gathered}
a\in\ran A,\quad b\in\ran B,\\
c\ge c_A+c_B,\\
c_A=a^\dagger A^+a,\quad c_B=b^\dagger B^+b .
\end{gathered}
\label{eq:exteriorSchur}
\end{equation}
Here, the superscript \(+\) denotes the Moore--Penrose generalized inverse.
Thus, \(T\) can be decomposed as
\begin{equation}
T=
\begin{pmatrix}A&0&a\\0&0&0\\a^\dagger&0&c_A\end{pmatrix}
+\begin{pmatrix}0&0&0\\0&B&b\\0&b^\dagger&c_B\end{pmatrix}
+(c-c_A-c_B)\proj{nm}.
\label{eq:exteriorsplit}
\end{equation}
The three terms on the right-hand side are all positive,
and their supports lie in \(\mathbb C^3\otimes m\),
\(n\otimes\mathbb C^3\), and \(\mathbb C(nm)\), respectively.
Each of these subspaces has local dimension one on at least one side,
so all three terms are separable, and hence \(T\) is separable.
This argument does not require the blocks to be invertible, and also allows them to be zero.
\end{proof}

\subsection{Existence and uniqueness of the least projection}

The candidate projectors in Eq.~\eqref{eq:candidates} commute with one another.
When compared by subspace inclusion, only \(R_A\) and \(R_B\) do not contain each other.
If both satisfy the admissibility conditions,
then \(\rho\) has no coherence terms among the following four product subspaces:
\(P_A\otimes P_B\), \(P_A\otimes P_B^\perp\),
\(P_A^\perp\otimes P_B\), and \(P_A^\perp\otimes P_B^\perp\).
Except for the first subspace, each of the remaining subspaces has dimension one on at least one side,
so the corresponding positive-operator blocks are all separable.
Therefore, \(R_0\) also satisfies the admissibility conditions.
It follows that the intersection of admissible subspaces is still admissible,
and hence there exists a unique least element \(R_*\) among the admissible projectors.

The local supports on both sides of \(C=R_*\rho R_*\) determine \(R_*\) exactly.
Indeed, if the factor of \(R_*\) on Alice's side is \(I\),
while \(C\) occupies only a two-dimensional subspace on that side,
then, by Eq.~\eqref{eq:negativeblock},
this support must correspond to \(P_A\).
Shrinking \(I\) to \(P_A\)
leaves the retained block and the separable exterior part unchanged,
yet yields a smaller admissible projector, contradicting the minimality of \(R_*\).
The same argument applies to Bob's side.

For a PPT entangled state, if there exists an admissible projector smaller than \(I\),
the corresponding product subspace is at most \(2\times3\).
The retained block after projection still satisfies the PPT condition,
and is therefore necessarily separable in this dimension~\cite{PPTcriterion}.
In this way, both the retained block and the exterior part are separable,
and the entire state \(\rho\) is also separable, contradicting entanglement.
Therefore, a PPT entangled state has only the identity projector as an admissible projector,
i.e., \(C(\rho)=\rho\) and \(w_*(\rho)=1\).
This completes the definition for all entangled input states.

\subsection{Proof of the canonical criterion}

Sufficiency is directly guaranteed by the flag protocol.
If \(C(\sigma)=WC(\rho)W^\dagger\),
then, by the preceding conclusion about local supports,
\(W\) also maps the least retained projector \(R_*(\rho)\) to \(R_*(\sigma)\).
Therefore, both Eq.~\eqref{eq:projector1} and Eq.~\eqref{eq:projector2}
hold, and the corresponding finite protocols realize conversion in both the forward and reverse directions.

To prove necessity, take the product unitary operator \(W\) relating the two retained blocks in the projector formulation,
and write the two states as
\(\rho=C_R+T_R\) and \(\sigma=WC_RW^\dagger+T'_{R'}\),
where \(R'=WRW^\dagger\).
If \(D\le R\) is an admissible projector of \(\rho\),
then let \(D'=WDW^\dagger\).
It is easy to verify that \(D'\) commutes with \(\sigma\), and that
\begin{equation}
\sigma-D'\sigma D'
=W\{R(\rho-D\rho D)R\}W^\dagger+T'_{R'}\in\SEP.
\label{eq:minimaltransport}
\end{equation}
Since \(R\) is a local product projector,
the compression map \(R(\,\cdot\,)R\) preserves separability,
so \(D'\) is also an admissible projector of \(\sigma\).
In particular, taking \(D=R_*(\rho)\) gives
\(R_*(\sigma)\le WR_*(\rho)W^\dagger\).
Conversely, repeating the same argument with \(W^\dagger\) yields the opposite inequality.
Thus
\begin{equation}
R_*(\sigma)=WR_*(\rho)W^\dagger,\qquad
C(\sigma)=WC(\rho)W^\dagger.
\label{eq:canonicaltransport}
\end{equation}
This completes the proof of Theorem~\ref{thm:criterion}.

Combined with Case~1 of Theorem~\ref{thm:classification},
this makes it possible to decide whether any two states belong to the same reversible equivalence class:
two separable states are always equivalent;
a separable state and an entangled state are not equivalent;
and two entangled states are decided by comparing the LU equivalence of their respective canonical retained blocks.
In particular, PPT states and NPT states cannot be mutually converted.

\section{Exact procedure and its scope}\label{sec:algorithm}

For an NPT input state, first perform the spectral decomposition of \(\rho^\Gamma\),
take its full negative part, and then use Eq.~\eqref{eq:negativeplanes}
to determine \(P_A,P_B\).
If both have rank two,
then compute \(d_0,d_A,d_B\) in Eq.~\eqref{eq:zerotests}.
By Lemma~\ref{lem:exterior}, the least retained projector is
\begin{equation}
R_*(\rho)=
\begin{cases}
R_0,&d_0=0,\\
R_A,&d_0>0,\ d_A=0,\\
R_B,&d_0>0,\ d_B=0,\\
I,&d_0,d_A,d_B>0.
\end{cases}
\label{eq:decision}
\end{equation}
By the earlier conclusion about the intersection of admissible subspaces,
the middle two cases cannot occur simultaneously.
If the local rank of the negative part is \((2,3)\),
then one only needs to choose between \(R_A\) and \(I\):
take \(R_A\) when \(d_A=0\), and \(I\) otherwise.
When the local rank is \((3,2)\), make the corresponding choice according to \(d_B\);
when it is \((3,3)\), directly take \(I\).
After constructing \(C(\rho)\), repeat this procedure for \(\sigma\),
and then determine whether the two unnormalized retained blocks are locally unitary equivalent.
These steps do not require solving the entanglement of formation (EOF)
or best separable approximation (BSA) optimization problems.

Before solving the LU equivalence equations,
one may first check the traces, ranks, squared Hilbert--Schmidt norms, global spectra, and reduced spectra on both sides of the two retained blocks.
If any one of these quantities differs, equivalence is immediately ruled out;
but agreement of these quantities is generally still not sufficient to establish equivalence.
If the retained block has rank one, then the spectrum of \(\Tr_B C\)
determines both its weight and its Schmidt coefficients,
so equality of the reduced spectra is a necessary and sufficient condition for LU equivalence.
For a \(2\times2\) mixed retained state,
after confirming that the weights are equal,
one can compare using the complete LU invariants of two-qubit states~\cite{Makhlin}.

If the normalized retained state has nondegenerate reduced spectra on both sides,
then after choosing local eigenbases in the same eigenvalue order,
the only remaining degrees of freedom are local phases.
In this case it suffices to check whether there exist phases such that
\begin{equation}
\eta'_{ij,kl}
=e^{i(\alpha_i-\alpha_k+\beta_j-\beta_l)}\eta_{ij,kl}.
\label{eq:phasetest}
\end{equation}
If the reduced spectra have degeneracies,
then arbitrary unitary transformations are still allowed within the degenerate eigenspaces,
and one cannot compare only the phases.
The general LU equivalence decision problem can be written as the problem of solving the following finite system of polynomial equations:
\begin{equation}
U^\dagger U=V^\dagger V=I,\qquad
(U\otimes V)C(\rho)=C(\sigma)(U\otimes V).
\label{eq:LUfeasibility}
\end{equation}
Splitting each matrix element into real and imaginary parts yields a system of real polynomial equations.
Candidate solutions obtained numerically can be checked by substituting back into the equations,
but failure of numerical optimization to find a solution does not amount to a proof that no solution exists.

For two known PPT entangled states,
reversible equivalence is exactly LU equivalence of the full density matrices.
If it has not yet been determined whether the two are separable,
LU equivalence is still a sufficient condition for mutual convertibility.
If LU equivalence does not hold, then they can be mutually converted only when both are separable.
Therefore, in this case one still needs to decide the separability of general two-qutrit states.

When the input data are all algebraic numbers,
the separability decision can be written as whether there exists a convex decomposition into at most 81 product projectors,
and then solved by real algebraic methods.
This paper does not claim that this step or the general LU polynomial system can be solved efficiently.
Likewise, exact zero testing requires exact data;
small residuals obtained from floating-point computation can provide only numerical evidence.
The above criterion addresses whether two states can be mutually converted,
and does not solve the general problem of one-way transformations between mixed states.

\section{Examples and an operational interpretation}
\label{sec:examples}

The following examples illustrate the roles played in the criterion by the weight of the retained block,
the normalized state, and the separability of the exterior part, respectively.
These examples also show that the negative part of the partial transpose cannot by itself determine the reversible equivalence class of the original state.
Let \(\Pi_x=\proj x\), and set
\begin{equation}
\begin{gathered}
\ket{\Phi^\pm}=\frac{\ket{00}\pm\ket{11}}{\sqrt2},\\
\ket{\Psi^\pm}=\frac{\ket{01}\pm\ket{10}}{\sqrt2},\\
P=\proj0+\proj1 .
\end{gathered}
\label{eq:Bellnotation}
\end{equation}
In what follows, all these projectors and states are regarded as operators and states in the two-qutrit space.

\subsection{A reversible pair with different global ranks}

First consider the pair of reversible states in the main text.
Their partial transposes have the same negative part:
\begin{equation}
(\rho^\Gamma)_-=(\sigma^\Gamma)_-=\tfrac3{10}\Pi_{\Psi^-}.
\end{equation}
Thus \(P_A=P_B=P\).
Both satisfy \(d_0=0\), and
\(C(\rho)=C(\sigma)=\tfrac35\Pi_{\Phi^+}\).
The forward and reverse protocols are those described for Eq.~\eqref{eq:examplemain} in the main text.
However, the global ranks of the two states are two and four, respectively,
so they cannot be transformed into each other by local unitary transformations.

\subsection{Weights and Schmidt coefficients both matter}

In addition to the above \(\rho\), take also
\begin{equation}
\begin{gathered}
\omega=\tfrac12\Pi_{\Phi^+}+\tfrac12\Pi_{22},\\
\xi=\tfrac35\Pi_\chi+\tfrac25\Pi_{22},\\
\ket\chi=\sqrt{\tfrac45}\ket{00}+\sqrt{\tfrac15}\ket{11}.
\end{gathered}
\end{equation}
All three states have the same least retained projector \(R_*=P\otimes P\),
and the retained blocks all have rank one.
The reduced spectra of the unnormalized retained blocks are respectively
\begin{equation}
\begin{gathered}
\rho:\ (3/10,3/10,0),\\
\omega:\ (1/4,1/4,0),\\
\xi:\ (12/25,3/25,0).
\end{gathered}
\end{equation}
Since the reduced spectra differ, the three states pairwise do not belong to the same reversible equivalence class.
Comparing \(\rho\) with \(\omega\) reveals the difference in retained weight;
comparing \(\rho\) with \(\xi\) reveals the difference in Schmidt coefficients when the weights are the same.
Therefore, even when the retained block corresponds to a pure state,
its weight and Schmidt coefficients are both information necessary for determining reversible equivalence.

\subsection{The same negative part does not determine the block}

Consider
\begin{equation}
\begin{gathered}
\eta_1=\tfrac35\Pi_{\Phi^+}+\tfrac25\Pi_{\Phi^-},\\
\eta_2=\tfrac35\Pi_{\Phi^+}+\tfrac15\Pi_{\Phi^-}
+\tfrac15\Pi_{\Psi^+},\\
\rho_i=\tfrac45\eta_i+\tfrac15\Pi_{22}.
\end{gathered}
\end{equation}
The partial transposes of the two states have the same negative part \(\tfrac2{25}\Pi_{\Psi^-}\),
and both satisfy \(d_0=0\) and \(w_*=4/5\).
However, the nonzero eigenvalues of the two retained blocks are respectively
\begin{equation}
\begin{gathered}
\operatorname{spec}_{\ne0} C(\rho_1)=(12/25,8/25),\\
\operatorname{spec}_{\ne0} C(\rho_2)=(12/25,4/25,4/25).
\end{gathered}
\end{equation}
Since the spectra differ, the two retained blocks cannot be locally unitary equivalent,
and therefore \(\rho_1\) and \(\rho_2\) cannot be mutually converted.
The squared Hilbert--Schmidt norms of their unnormalized retained blocks are \(208/625\) and \(176/625\), respectively,
from which the same conclusion can also be drawn.
Thus, the same negative part and retained weight are still insufficient to determine the reversible equivalence class;
the state inside the retained block must also be included in the comparison.
These two examples also show that the negative part having rank one does not imply that the retained state is pure.

\subsection{The exterior coherence test is necessary}

Let \(\ket s=(\ket{02}+\ket{20})/\sqrt2\), and consider
\begin{equation}
u=\tfrac35\Pi_{\Phi^+}+\tfrac15\Pi_s+\tfrac15\Pi_{22},\qquad
v=\tfrac35\Pi_{\Phi^+}+\tfrac15\Pi_{02}+\tfrac15\Pi_{20}.
\end{equation}
The partial transposes of the two states have the same negative part \(\tfrac3{10}\Pi_{\Psi^-}\).
For \(u\), the block of the partial transpose on \(\{\ket{00},\ket{22}\}\) is
\(\left(\begin{smallmatrix}3/10&1/10\\1/10&1/5\end{smallmatrix}\right)>0\),
and does not produce any additional negative eigenvalue.
Although both states commute with \(R_0=P\otimes P\),
their coherence terms between the two wings are different:
\begin{equation}
X_u=\tfrac1{10}\ket{02}\bra{20}\ne0,\qquad X_v=0 .
\end{equation}
Therefore,
\begin{equation}
\begin{gathered}
(d_0,d_A,d_B)_u=(1/100,1/100,1/100),\\
(d_0,d_A,d_B)_v=(0,0,0).
\end{gathered}
\end{equation}
For \(u\), the entire state must be retained, i.e., \(C(u)=u\), and the retained weight is one;
for \(v\), it suffices to retain \(C(v)=\tfrac35\Pi_{\Phi^+}\),
with weight \(3/5\).
The retained weights of the two states are different, so they cannot be mutually converted.
This example shows that even if a state commutes with \(R_0\),
one cannot conclude from this that the exterior part can be replaced:
coherence between the two wings may still make it entangled.
The second term in \(d_0\) is precisely there to test this.

\subsection{What the canonical weight measures}

Let \(\eta_*=C(\rho)/w_*\).
Let \(E\) be a convex LOCC strong monotone that vanishes on separable states.
Applying convexity to the canonical decomposition gives \(E(\rho)\le w_*E(\eta_*)\);
reading the local flag and then using strong monotonicity gives the reverse inequality.
Therefore,
\begin{equation}
E(\rho)=w_*E(\eta_*).
\label{eq:weightedE}
\end{equation}
This is precisely the application of the flag property~\cite{Flags} to the canonical decomposition.
Since \(E_F(\eta_*)>0\),
this also determines the maximum success probability of exactly obtaining \(\eta_*\) from \(\rho\) via LOCC:
\begin{equation}
p_{\max}(\rho\longrightarrow\eta_*)=w_*.
\label{eq:extraction}
\end{equation}
Measuring the local flag achieves the success probability \(w_*\).
On the other hand, the strong monotonicity of \(E_F\) and Eq.~\eqref{eq:weightedE}
guarantee that the success probability of any protocol cannot exceed \(w_*\).
Here one may merge several measurement outcomes into a success event;
it is only required that the conditional state after merging equals \(\eta_*\),
and it is not required that every success branch output the same state.

\end{document}